\def\ispreprint{}
\documentclass{article}

\usepackage{microtype}
\usepackage{graphicx}
\usepackage{subcaption}
\usepackage{booktabs} 
\usepackage{mdframed}
\usepackage{algorithm}
\usepackage{algorithmic}
\usepackage{enumitem}
\usepackage{booktabs}
\usepackage{multirow}
\usepackage{xspace}
\usepackage{pifont}
\usepackage{hyperref}

\setlist[itemize]{leftmargin=*, itemsep=0pt, label=\large\textbullet}
\setlist[enumerate]{leftmargin=*, itemsep=0pt}

\newcommand\name{\textsc{Immaculate}\xspace}
\newcommand\parens[1]{\left(#1\right)}

\newcommand\braces[1]{\left\{#1\right\}}
\newcommand\angles[1]{\left\langle #1\right\rangle}

\usepackage{xcolor}

\renewcommand{\algorithmicwhile}{\textbf{upon}}

\usepackage{icml2026}
\AtBeginDocument{
  \setlength{\parskip}{2.5pt}
}

\usepackage{amsmath}
\usepackage{amssymb}
\usepackage{mathtools}
\usepackage{amsthm}
\usepackage{adjustbox}
\usepackage{enumitem}

\usepackage[capitalize,noabbrev]{cleveref}

\theoremstyle{plain}
\newtheorem{theorem}{Theorem}[section]
\newtheorem{proposition}[theorem]{Proposition}

\theoremstyle{definition}
\newtheorem{definition}[theorem]{Definition}

\theoremstyle{remark}

\newcommand{\xmark}{\ding{55}}%
\newcommand{\cmark}{\ding{51}}%
\newcommand{\mycheck}{{\color{green}\cmark}}
\newcommand{\mycross}{{\color{red}\xmark}}

\setlist[itemize]{
  leftmargin=*,
  topsep=0pt,
  itemsep=0pt,
  parsep=0pt,
  partopsep=0pt
}

\setlist[enumerate]{
  leftmargin=*,
  topsep=0pt,
  itemsep=0pt,
  parsep=0pt,
  partopsep=0pt
}

\usepackage[textsize=tiny]{todonotes}

\begin{document}

\twocolumn[
  \icmltitle{IMMACULATE: A Practical LLM Auditing Framework via Verifiable Computation}



  \icmlsetsymbol{equal}{*}

  \begin{icmlauthorlist}
    \icmlauthor{Yanpei Guo}{nus}
    \icmlauthor{Wenjie Qu}{nus}
    \icmlauthor{Linyu Wu}{nus}
    \icmlauthor{Shengfang Zhai}{nus}
    \icmlauthor{Lionel Z. WANG}{ntu}
    \icmlauthor{Ming Xu}{nus}
    \icmlauthor{Yue Liu}{nus}
    \icmlauthor{Binhang Yuan}{ind}
    \icmlauthor{Dawn Song}{ucb}
    \icmlauthor{Jiaheng Zhang}{nus}
  \end{icmlauthorlist}

  \icmlaffiliation{nus}{National University of Singapore}
  \icmlaffiliation{ntu}{Nanyang Technological University}
  \icmlaffiliation{ind}{Independent Researcher}
  \icmlaffiliation{ucb}{University of California, Berkeley}

  \icmlcorrespondingauthor{Jiaheng Zhang}{jhzhang@nus.edu.sg}

  \icmlkeywords{Machine Learning, ICML}

  \vskip 0.3in
]



\printAffiliationsAndNotice{}  

\begin{abstract}
Commercial large language models are typically deployed as black-box API services, requiring users to trust providers to execute inference correctly and report token usage honestly.
We present IMMACULATE, a practical auditing framework that detects economically motivated deviations-such as model substitution, quantization abuse, and token overbilling-without trusted hardware or access to model internals.
IMMACULATE selectively audits a small fraction of requests using verifiable computation, achieving strong detection guarantees while amortizing cryptographic overhead.
Experiments on dense and MoE models show that IMMACULATE reliably distinguishes benign and malicious executions with under 1\% throughput overhead. 
Our code is published at \url{https://github.com/guo-yanpei/Immaculate}.
\end{abstract}

\section{Introduction}

Large language models (LLMs) have become a critical computational substrate across a wide range of applications, enabling high-quality text generation~\cite{brown2020language, chowdhery2023palm}, conversational agents~\cite{adamopoulou2020overview}, code synthesis~\cite{chen2021evaluating}, and reasoning-intensive workloads~\cite{wei2022chain}. 
Today, most frontier LLMs including OpenAI~\cite{openai_api}, Google~\cite{google_vertex_ai}, and Anthropic~\cite{anthropic_api} are deployed through \emph{API-based inference services}, where users interact with proprietary models via black-box endpoints and are charged based on reported token usage and model specifications.


Despite its success, the black-box API model introduces a fundamental \emph{trust asymmetry}. 
Users lack any visibility into the deployed model or its internal execution, making it hard for users to audit whether the providers faithfully executed the claimed inference procedure.  
From an economic perspective, providers may have incentives to silently deviate from the advertised protocol—for example, by substituting a distilled model or applying aggressive quantization to reduce computational cost. 
Indeed, anecdotal reports of quality degradation under unchanged model claims have appeared in public forums, highlighting an emerging accountability gap in LLM service provision~\cite{gpt_lazy1, gpt_lazy2}.

Compared to model substitution, \emph{token overbilling} poses a more severe integrity threat~\cite{sun2025coin}.
For reasoning-intensive workloads, API providers (e.g., OpenAI) hide intermediate chain-of-thought tokens while still billing for full internal token usage~\cite{openai_api}.
As token usage is a dominant cost factor for downstream applications and autonomous agents~\cite{chen2023frugalgpt,mohammadi2025evaluation}, the absence of verifiable billing guarantees undermines consumer trust.
Unlike model substitution, token misreporting leaves no observable signal in the output, making detection impossible under standard black-box access.

These challenges highlight the growing importance of \emph{auditing black-box LLM services}.
As LLM APIs become critical infrastructure for downstream applications and autonomous agents, users must be able to verify that providers execute inference honestly—without gaining access to proprietary models or internal states.
However, auditing black-box LLM services is inherently difficult due to competing system requirements: 
\begin{itemize}
\item \textbf{Economy}: the auditing mechanism should introduce minimal computational and monetary overhead for benign LLM service providers.
\item \textbf{Proprietorship}: the LLM vendor’s model architecture, parameters, and internal inference states should remain confidential.
\item \textbf{Robustness}: the framework should offer provable guarantees for detecting protocol deviations, without falsely flagging honest servers.
\end{itemize}

A representative approach~\cite{cai2025you} enforces runtime integrity using trusted execution environments (TEEs), particularly GPU-based TEEs~\cite{nvidia2023hoppercc}.
While TEEs provide strong integrity guarantees with modest overhead ($\approx 20\%$ on NVIDIA H100~\cite{tan2025pipellm}), their economic viability is limited: GPU-TEE support is restricted to specific NVIDIA hardware, and large-scale adoption of GPU-TEEs would require costly redesign of existing inference infrastructure (e.g., Google TPUs~\cite{jouppi2017datacenter}, AWS Inferentia~\cite{aws_inferentia}, Ascend GPUs~\cite{zhou2025accelerating}).

An alternative line of work relies on cryptographic proofs to attest correct inference execution~\cite{chen2024zkml}.
However, such approaches incur substantial prover overhead and are fundamentally limited to deterministic computations.
Existing systems (e.g., zkLLM~\cite{sun2024zkllm}, zkGPT~\cite{qu2025zkgpt}) therefore restrict inference to integer-only arithmetic, which is inefficient on modern GPUs, rendering these approaches impractical for real-world LLM services~\cite{cai2025you,gao2024model}.

\paragraph{Our approach.}
In this work, we present \name, a practical auditing framework for black-box LLM services that following the line of proof-based approaches.  
Compared to prior proof-based schemes, \name eliminates the need for integer-only inference and remains compatible with existing inference infrastructures, while incurring only minor proving overhead.

The key insight behind cost-saving is economic: a rational service provider can only benefit from misbehavior by deviating from the claimed inference or billing protocol on \emph{a non-negligible fraction of} requests.  
Therefore, instead of verifying every request, the server can be considered benign if it can successfully prove its execution on a small, random subset of requests.  
As a result, an auditor can detect large-scale deviations with overwhelming probability by issuing a small number of random queries that mimic normal user requests and requiring the server to prove their execution.
Consequently, the relative overhead of proof generation is negligible compared to the cost of serving billions of inferences per day. 
This idea is visualized in Figure \ref{fig:framework}.

To address the inherent non-determinism of floating-point execution on modern accelerators, we introduce the logit distance distribution (LDD) as a quantitative indicator of approximation fidelity between a returned response and a claimed model.
Rather than proving exact LLM execution—which is infeasible under numerical non-determinism—we require the server to prove the computation of the LDD.
The resulting distribution serves as a \emph{verifiable audit footprint} that captures systematic deviations from the claimed inference procedure while remaining stable under benign numerical noise.
When aggregated over a large number of requests, LDDs induced by benign execution and malicious deviation become discriminative.
This design enables sound verification of LLM execution without requiring bitwise reproducibility.

\begin{figure}
    \centering
    \includegraphics[width=0.6\linewidth]{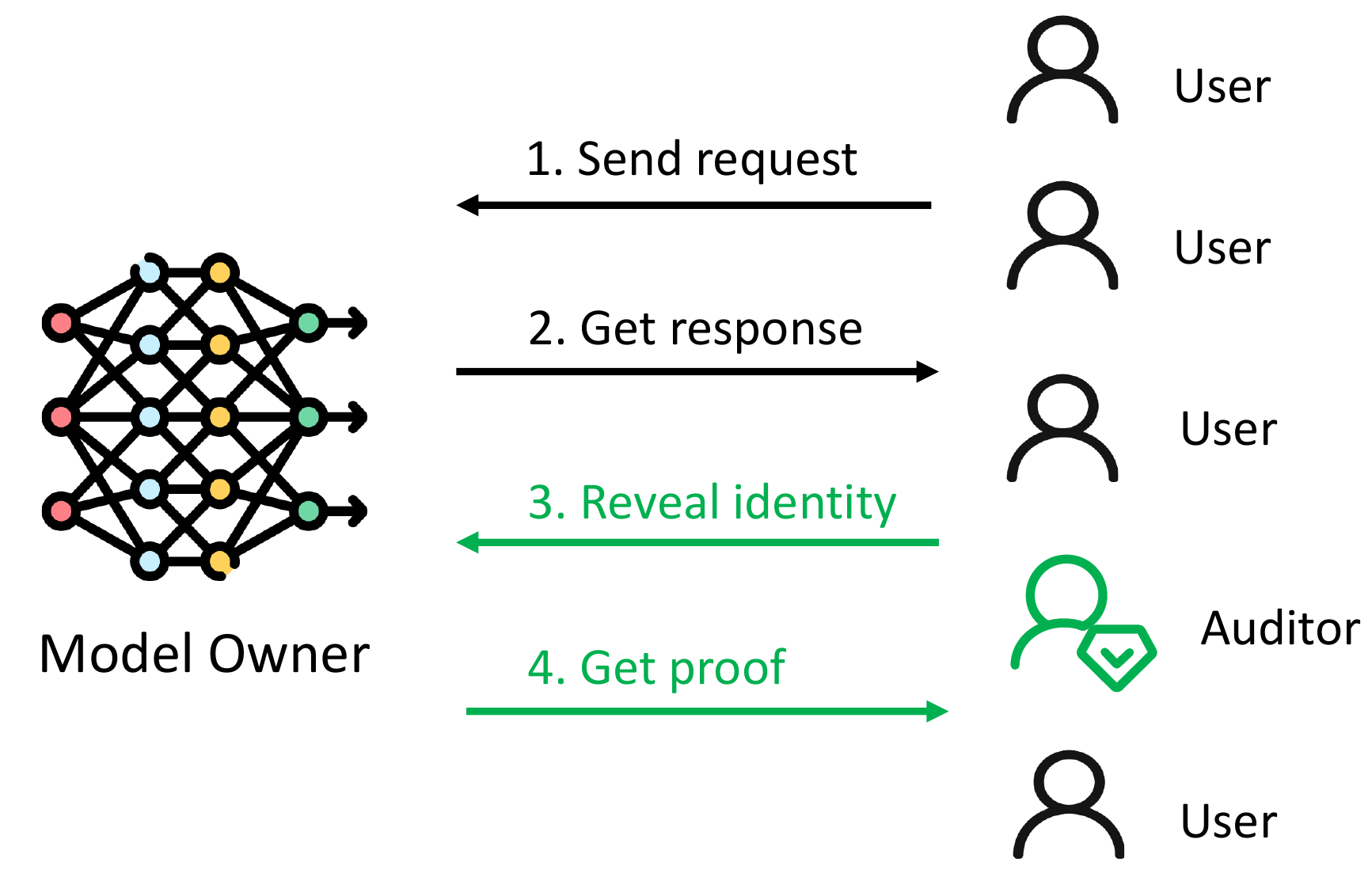}
    \caption{The auditor sends random requests like other users and requires the model owner to prove the responses after they are received.}
    \label{fig:framework}
\end{figure}

By combining probabilistic auditing and proving LDD on audited requests, \name provides economical, strong, provable guarantees for LLM service integrity.
The contributions can be summarized as follows:
\begin{itemize}
    \item We propose \name, a practical auditing framework for black-box LLM services that provides strong integrity guarantees without relying on GPU-based trusted execution environments.
    
    \item We introduce the \emph{logit distance distribution} (LDD), a fidelity metric that quantifies the approximation gap between a deployed execution and a claimed full-precision model, enabling verification under numerical non-determinism.
    
    \item We demonstrate that \name is economical and compatible with existing large-scale LLM deployments. Our prototype implementation built on \texttt{vLLM} incurs only $\sim$1\% overhead.
    Our code is published at \url{https://github.com/guo-yanpei/Immaculate}.
\end{itemize}

\section{Preliminary and Related Work}

\paragraph{Notation.}
For a function $f(\cdot)$, the assignment $y := f(x)$ denotes a \emph{deterministic}
map whose output is uniquely determined by its input.
The assignment $y \leftarrow g(x)$ denotes a \emph{non-deterministic} map, typically
reflecting hardware-dependent variability such as GPU parallelism.
We use $[n]$ to denote set $\braces{1, 2, \cdots, n}$.

\subsection{Cryptographic Commitment}

Cryptographic commitment schemes~\cite{merkle2019protocols} allow a party to commit to a message and later reveal it, while guaranteeing binding and hiding. In practice, commitments can be instantiated using cryptographic hash functions, e.g., $c = \mathsf{Hash}(m || r)$.

\subsection{Verifiable Computation}

Verifiable computation (VC)~\cite{gennaro2013quadratic} enables a powerful prover to convince a lightweight verifier that a \emph{deterministic} public function $f(\mathbf{x}, \mathbf{w}) = \mathbf{y}$ was executed correctly, without revealing the private witness $\mathbf{w}$.
Compared to re-execution, VC provides sublinear (often polylogarithmic) verification cost and witness privacy.
Formally, VC consists of a prover--verifier pair $(\mathsf{Prove}, \mathsf{Vrfy})$ satisfying completeness, soundness, and zero-knowledge.
VC can be instantiated via Trusted Execution Environments (TEEs)~\cite{sabt2015trusted} or zero-knowledge proofs (ZKPs)~\cite{goldwasser2019knowledge}:
TEEs offer near-native performance under hardware trust assumptions, while ZKPs provide stronger cryptographic guarantees at substantially higher cost.

\subsection{Related Works} \label{sec:related}

\begin{table}[!ht]
    \caption{The comparison between our work and prior works.}
    \label{tab:literature}
    \begin{adjustbox}{max width=1\linewidth}

    \begin{tabular}{@{}cclclc@{}}
    \toprule
        & \textbf{Robustness} &  & \textbf{Infra-agnostic} &  & \textbf{Efficiency} \\ \midrule
    GPU TEE        & \mycheck             &  & \mycross                 &  & \mycheck           \\
    Cryptography    & \mycheck             &  & \mycheck                 &  & \mycross           \\
    Empirical         & \mycross             &  & \mycheck                 &  & \mycheck           \\
    \midrule
    Our Approach                                                               & \mycheck             &  & \mycheck                 &  & \mycheck           \\ \bottomrule
    \end{tabular}

    \end{adjustbox}
  
\end{table}

Existing LLM auditing approaches can be broadly categorized into three paradigms: empirical methods, GPU-TEE–based verification, and cryptographic verification.
Table~\ref{tab:literature} summarizes their trade-offs in terms of robustness, hardware requirements, and efficiency.

GPU-TEE–based approaches provide strong robustness guarantees but rely on specialized infrastructure, which introduces considerable hardware procurement costs.
Cryptographic approaches offer strong security guarantees without trusted hardware; however, they incur substantial proof-generation cost and are typically restricted to integer-only model executions, making them inefficient on modern AI-accelerated hardware.

Empirical auditing methods, which operate without trusted hardware or cryptographic primitives, are computationally efficient and easily deployable. 
Nevertheless, they generally lack formal guarantees of completeness and soundness, and are typically vulnerable to adaptive attacks such as aggressive quantization~\cite{cai2025you} or token overbilling~\cite{sun2025coin}.

In contrast, our approach achieves a favorable balance across all three dimensions, combining robustness with hardware compatibility and practical efficiency.
A more detailed discussion of existing auditing schemes is provided in Appendix~\ref{app:related}.

\section{Threat Model}

\paragraph{System setting.}
Our auditing framework involves two parties: a cloud-based LLM inference server $\mathsf{Srv}$ and a trusted auditor $\mathsf{Adt}$.
The server claims to perform inference using a specified model $\mathcal{M}_{\theta}$ with a certain precision.
For each query $\vec{x}$, $\mathsf{Srv}$ returns an output sequence $\vec{y}$ along with a reported token count $T$.
The goal of $\mathsf{Adt}$ is to act as a normal user by submitting random queries to $\mathsf{Srv}$ and verifying whether $\mathsf{Srv}$ processes the requests faithfully.


\paragraph{Threat model.}
We assume $\mathsf{Srv}$ is rational and economically motivated, and may deviate from the claimed protocol to reduce computation cost or increase billing. 
We exclude attacks that increase computational cost and manage to serve users with an alternative model. 
A malicious server is assumed to misbehave on at least $10\%$ of queries.
We assume $\mathsf{Adt}$ always behaves honestly, and $\mathsf{Srv}$ cannot distinguish the requests from $\mathsf{Adt}$ and from normal user.

\paragraph{Auditing goal.}
An auditing scheme must distinguish between the following two types of servers:
\begin{enumerate}
    \item \textbf{Honest execution:} $\mathsf{Srv}(\vec{x}) = \mathcal{M}_{\theta}(\vec{x})$ for all $\vec{x}$.
    \item \textbf{$\alpha$-dishonest execution:}
    \[
        \Pr_{\vec{x}}[\mathsf{Srv}(\vec{x}) \neq \mathcal{M}_{\theta}(\vec{x})] \ge \alpha.
    \]
\end{enumerate}

The auditing framework outputs \textsf{ACCEPT} or \textsf{REJECT} and satisfies:
\begin{itemize}
    \item \textbf{Completeness:} An honest server is accepted with overwhelming probability.
    \item \textbf{Soundness:} An $\alpha$-dishonest server is rejected with high probability.
    \item \textbf{Efficiency:} Auditing introduces only marginal overhead beyond standard inference.
    \item \textbf{Privacy:} The model's parameters and architecture remain hidden from all parties except the model server.
    \item \textbf{Generality:} The scheme applies across various LLM architectures and hardware platforms.
\end{itemize}

In our paper, following prior works~\cite{cai2025you, sun2025coin}, we consider the following integrity deviation attacks:
\begin{itemize}
    \item \textbf{Model substitution.}
The server replaces the claimed model with a lower-cost alternative.
\item \textbf{Aggressive quantization.}
The server executes the claimed model architecture but uses a lower-precision arithmetic format than advertised (e.g., FP8 instead of BF16). 
\item \textbf{Token overreporting.}
The server executes all recurrent hybrid steps faithfully but manipulates the output reconstruction function $D$ to inflate the reported token count $T$ beyond actual usage.
\end{itemize}

\section{Logit Distance Distribution}

A central challenge in auditing LLM inference is the absence of a well-defined ``ground truth" execution.
Due to GPU-level non-determinism and numerical instability in finite-precision arithmetic, benign executions can even produce different outputs.

Our key insight is that, for each LLM inference task, one can define an idealized and deterministic reference execution under full-precision arithmetic. 
Faithfulness can then be assessed by measuring the \emph{approximation fidelity} between the observed execution and this ideal reference. 
Intuitively, if the discrepancy between claimed model and executed model arise solely from numerical error, their distance would exhibit a stable and predictable distribution. 
In contrast, malicious behaviors—such as model substitution or aggressive quantization—introduce structured deviations that generate significantly stronger and easily distinguishable error signals.   
We hence try to use such distance distribution to quantize the approximation fidelity.

However, a central challenge lies in the fact that the discrete steps in LLMs can significantly amplify numerical errors.  
For example, since LLMs generate outputs auto-regressively, even a small numerical error that leads to a different token selection can cause the subsequent inference sequence to fully diverge from the reference.  
Therefore, a mechanism is needed to align the observed inference with the reference.

In Section \ref{sec:hybrid-compute}, we introduce an abstraction of general LLMs, which serves as the foundation for our alignment technique.  
In Section \ref{sec:ldd}, we present the \emph{logit distance distribution} (LDD) as a metric for quantifying approximation fidelity, which relies on our alignment technique.



\subsection{Hybrid Computation Model: LLM Abstraction} \label{sec:hybrid-compute}

Let $\mathcal{X} \subseteq \mathbb{R}^{+}$ denote the continuous hidden-state space, $\mathcal{D}$ a finite discrete decision set (e.g., vocabulary symbols or expert indices), $\theta$ the model parameter space, and $\Gamma$ the token vocabulary. Given a prompt $\vec{x} \in \Gamma^{+}$, a large language model (LLM) autoregressively produces an output sequence $\vec{y} \in \Gamma^{+}$.

\paragraph{Hybrid computation model.}
We abstract LLM inference as a hybrid continuous--discrete process
\[
\mathcal{M}_{\theta} := (E_{\theta}, F_{\theta}, S, G_{\theta}, D),
\]
which evolves through a sequence of transitions combining differentiable neural computation and discrete control-flow decisions. Concretely, inference proceeds as follows.

\begin{itemize}
    \item \underline{\textit{Initial embedding.}}  
    The input prompt $\vec{x}$ is mapped into an initial continuous hidden state:
    \[
        h_0 \leftarrow E_{\theta}(\vec{x}), 
        \qquad E_{\theta} : \Gamma^{+} \rightarrow \mathcal{X}.
    \]

    \item \underline{\textit{Recurrent hybrid step.}}  
    For each generation step $i = 1,2,\dots$, the model executes:
    \begin{itemize}
        \item \textbf{Continuous transformation.}  
        Differentiable neural operators (e.g., attention, MLPs, normalization) map the current state $h_{i-1}$ to an intermediate state $\tilde{h}_i$ and a decision vector $\ell_i$ (e.g., logits):
        \[
            (\tilde{h}_i, \ell_i) \leftarrow F_{\theta}(h_{i-1}), 
            \qquad F_{\theta} : \mathcal{X} \rightarrow \mathcal{X} \times \mathcal{X}.
        \]

        \item \textbf{Discrete decision.}  
        A discrete choice $d_i \in \mathcal{D}$ (such as token sampling or MoE routing) is selected according to $\ell_i$ and a randomness source $r_i$:
        \[
            d_i := S(\ell_i, r_i), 
            \qquad S : \mathcal{X} \times \mathcal{R} \rightarrow \mathcal{D}.
        \]

        \item \textbf{State update.}  
        The discrete decision is injected back into the continuous computation stream:
        \[
            h_i \leftarrow G_{\theta}(\tilde{h}_i, d_i), 
            \qquad G_{\theta} : \mathcal{X} \times \mathcal{D} \rightarrow \mathcal{X}.
        \]
    \end{itemize}

    \item \underline{\textit{Output reconstruction.}}  
    After $N$ recurrent steps, the output sequence and reported token count are reconstructed:
    \[
        (\vec{y}, T) := D(\{d_1, d_2, \dots, d_N\}).
    \]
\end{itemize}

The hybrid computation model distinguishes continuous transformations from discrete decisions in two key respects:
\begin{itemize}
\item Continuous transformations are subject to runtime non-determinism, whereas discrete decisions are fully determined by their inputs.
\item Continuous outputs vary smoothly under small input perturbations, while discrete decisions can change abruptly in response to the same perturbations, leading to different control-flow paths.
\end{itemize}


\paragraph{Ideal and approximate execution.}
We define the \emph{full-precision model} $\mathcal{M}^{\star}$ as the idealized execution of the above hybrid process where all continuous transformations are evaluated over the real numbers $\mathbb{R}$. 
In contrast, a deployed implementation $\mathcal{M}$ operates using finite-precision arithmetic (e.g., FP16, BF16, FP8), yielding a numerical approximation of $\mathcal{M}^{\star}$. 
A central goal is to quantify the fidelity of $\mathcal{M}$ relative to $\mathcal{M}^{\star}$.

\paragraph{Control-flow alignment.}

A crucial challenge in defining and measuring \emph{approximation fidelity} is that $\mathcal{M}$ and $\mathcal{M}^{\star}$ may follow completely different control flows due to tiny numerical differences.
This situation can arise when the two models make different discrete decisions.
One way to address this issue is to enforce $\mathcal{M}^{\star}$ to choose the same discrete decisions as $\mathcal{M}$ at every step.
This approach is visualized in Figure \ref{fig:continuity}.

\begin{figure}
    \centering
    \includegraphics[width=0.9\linewidth]{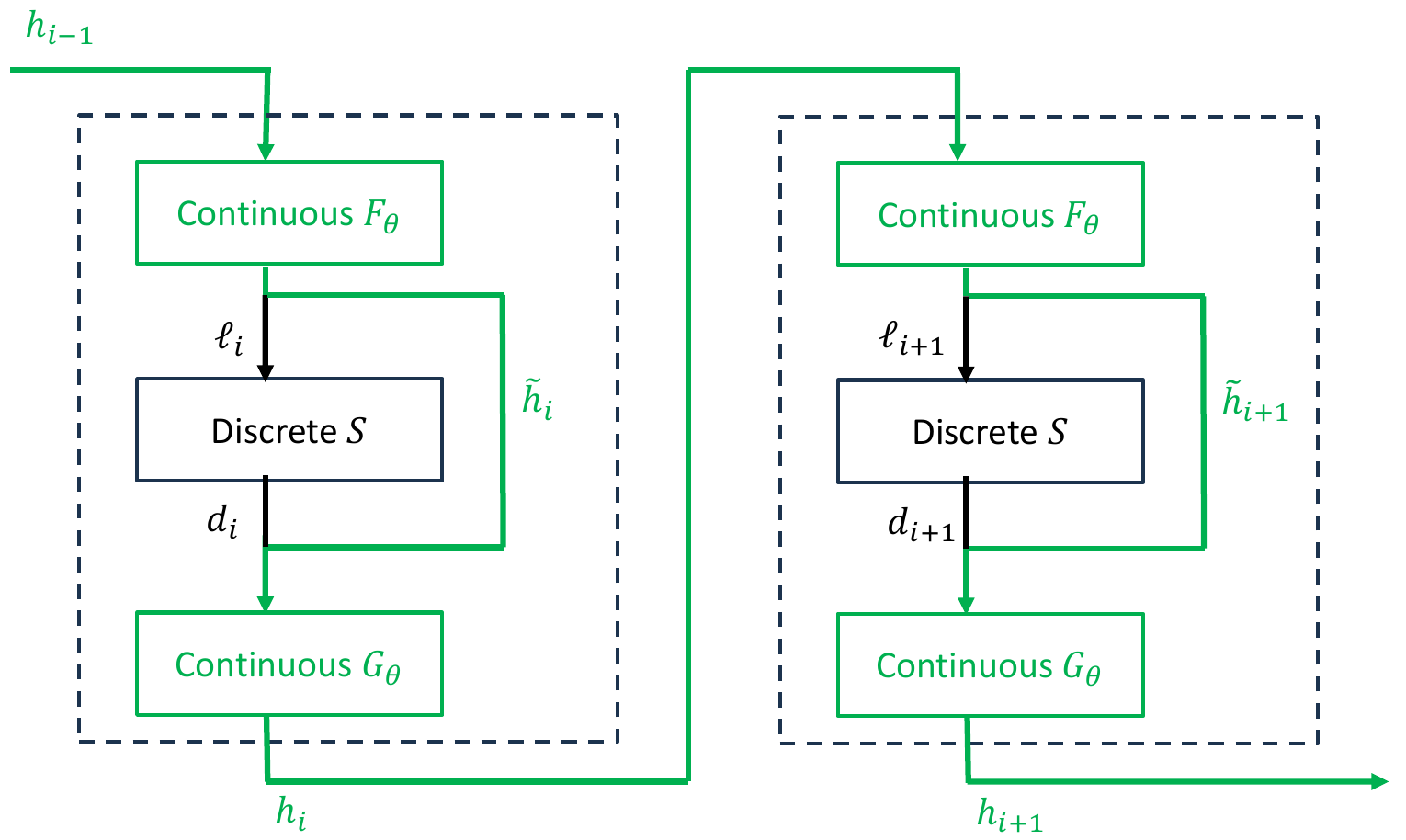}
    \caption{When fixing discrete selections ${d_i}$, the entire inference (green workflow) can be viewed as a continuous computation, as discrete selection no longer introduces branching uncertainty.}
    \label{fig:continuity}
\end{figure}


\subsection{Logit Distance Distribution} \label{sec:ldd}

However, this enforcement introduces a new issue.
For each enforced decision $d_i$, we must evaluate the ``proximity" between the logits $\ell_i^{\star}$ and the decision $d_i$.
If a decision is far inconsistent with the underlying state (i.e., the distance is too large), such enforcement should be flagged as misbehavior.

Measuring the fidelity between logits $\ell_i^{\star}$ and the discrete choice $d_i$ is challenging, as it involves comparing two inherently different entities: continuous values and discrete outcomes.  
Instead, we compare the distance between the \emph{inputs} to the discrete selection function, namely the logits $\ell_i$ and $\ell_i^{\star}$.  
Since the discrete selection function is deterministic, recording the logits is also sufficient to reproduce the decision $d_i$.

Therefore, validating the proximity of a discrete selection can be transformed into checking distance between $\ell_i$ and $\ell^{\star}_i$.
Any misbehavior tends to make $\ell_i$ and $\ell^{\star}_i$ far.
Hence, the approximation fidelity can be measured by logit distance distribution, formally defined below.

\begin{definition}[Logit Distance Distribution] \label{def:ldd}
Conditioned on an identical discrete decisions, let $\{\ell_i\}$ denote the logits produced by a deployed model execution, and let $\{\ell_i^{\star}\}$ denote the corresponding logits produced by the full-precision model.  
The \emph{logit distance distribution} (LDD) is defined as the distribution of $\{\mathsf{Dis}(\ell_i, \ell_i^{\star})\}$, where $\mathsf{Dis}(\cdot,\cdot)$ denotes a distance metric between two logit vectors, like KL divergence or total variance (TV) distance. 
\end{definition}





The following propositions characterizes the statistical signatures of three existing attacks studied in previous works~\cite{cai2025you, sun2025coin}:

\begin{proposition}
\label{prop:logit_bias_variance}
Model substitution introduces a \emph{systematic bias} in logit outputs. As a result, model substitution typically yields substantially larger LDDs. 
\end{proposition}

\begin{proposition}
\label{prop:logit_bias_variance_numerical}
Reducing numerical precision manifests as increased variance in logit deviations. 
Thus, precision reduction causes a progressive broadening of the LDD distribution.
\end{proposition}

Both propositions are proved in Appendix \ref{app:logit-bias}.

\begin{proposition} \label{prop:token_report}
Token overreporting is a special case of model substitution under the hybrid computation model.
\end{proposition}

A proof sketch is token overreporting can be viewed as adding some dummy recurrent hybrid steps to the model.
A complete proof is presented at Appendix \ref{app:overbill}. 
\section{Immaculate: An Auditing Framework for LLM Execution}

In this section, we present \name, a practical auditing framework for detecting $\alpha$-dishonest execution of large language models (LLMs) by an untrusted service provider.
At a high level, \name adopts randomized auditing to reduce proving cost, and combines LDD with verifiable computation to solve intrinsic numerical non-determinism.



\subsection{Randomized Auditing} \label{sec:random}

Our key observation for significantly reducing the proving cost is that it is unnecessary to prove every query. 
Instead, the auditor can adopt \emph{randomized auditing} by proving only a small random subset of queries. 
If a server deviates on an $\alpha$-fraction of requests, such behavior can be detected with overwhelming probability from this subset. 
Importantly, the required number of audited queries depends only on $\alpha$ and the desired confidence level, rather than on the total service volume.

\paragraph{Example.}
Assume a method can detect malicious responses with zero false positive and only $1\%$ detection rate.
If the server cheats on a fraction $\alpha = 0.1$ of requests, then achieving an overall detection probability of $95\%$ (i.e., evasion probability $\eta = 5\%$)
requires only
\[
    N = \frac{\log \eta}{\log (1 - \alpha \cdot 1\%)} \;\approx\; 3{,}000
\]
audited queries.
Since a production LLM service may process billions of requests per day, the cost
of proving only a few thousand queries is amortized to negligible overhead.

\subsection{Reproducibility via Discrete-State Commitments}

Exact reproduction of LLM execution is infeasible due to unavoidable numerical non-determinism.
Instead of proving bitwise equivalence, \name verifies that the execution remains within an admissible numerical deviation from a reference full-precision model, using the LDD metric.

At initialization, the model owner publishes a cryptographic commitment to the claimed \emph{full-precision} model $\mathcal{M}_{\theta^{\star}}$.
During inference, the server runs finite-precision model $\mathcal{M}_{\theta}$, recording and committing to $\{\ell_i\}$ at each discrete decision.
Importantly, these committed logits enables reproducing every discrete selection result.

For an audited query, server proves the LDD via VC, using the committed runtime logits $\braces{\ell_i}$.
Concretely, the VC proof establishes that:
\begin{enumerate}
    \item discrete decisions are derived from the committed logits $\{\ell_i\}$;
    \item continuous transformations are computed using $\mathcal{M}_{\theta^{\star}}$, obtaining $\braces{\ell^{\star}_i}$;
    \item the final output $(\vec{y}, T)$ is consistent with the discrete decisions;
    \item distribution of distance between $\{\ell_i\}$ and $\braces{\ell^{\star}_i}$ is output.
\end{enumerate}

\paragraph{Optimization: Top-$K$ distance.}
As top-$K$ selection is the most common discrete selection method in LLMs, we design a specialized top-$K$ distance metric, allowing the server to cache and commit only the $K$ selected indices instead of the full logits vector.  
The detailed method is provided in Appendix \ref{app:topk-dis}.

\begin{figure}
    \centering
    \includegraphics[width=0.6\linewidth]{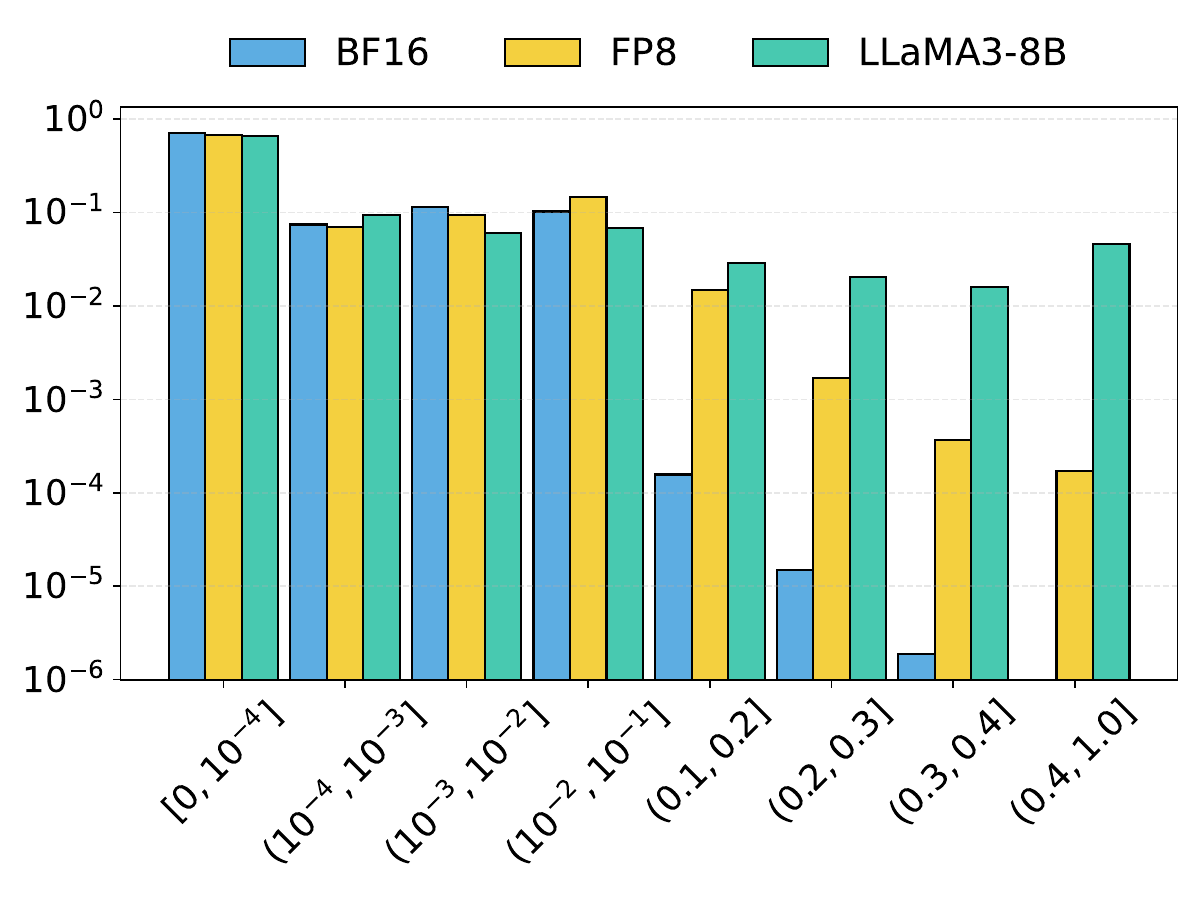}
    \caption{Logit TV-distance distribution of LLaMA3-70B. Probabilities are displayed on a logarithmic y-axis to better capture the tail behavior.}
    \label{fig:global_tv:llama}
\end{figure}

\subsection{Utilizing LDD in Practice.} \label{sec:ldd-use}

A central goal of our scheme is to produce observable and verifiable auditing footprints—concrete statistical evidence that can be independently examined to assess service integrity. 
The LDD serves as such an auditing footprint. 
There are many possible ways to utilize LDD for integrity auditing. 
In this work, we present LDD as a clear metric to distinguish between benign and malicious responses, along with some simple yet effective approaches for leveraging LDD in practice.

Empirically, we test LDD over thousands of requests, and present a representative result in Figure \ref{fig:global_tv}.
It aggregates benign executions and malicious executions over thousands of requests, and find they induce well-separated distance distributions.
We can observe that a particularly discriminative statistic is the tail probability, e.g. $\Pr[\mathsf{TV}(\ell_i, \ell^{\star}_i) > 0.1]$. 

Based on this observation, we can detect malicious requests by evaluating the density of ``tailing logits".
For a request involving $T$ discrete steps, we define
\[
p(t_1) := \frac{1}{T} \left| \left\{ i \in [T] : \mathsf{TV}(\ell_i, \ell^{\star}_i) > t_1 \right\} \right|.
\]
A response is flagged as malicious (e.g. computed vis FP8) if $p(t_1) > t_2$, where thresholds $t_1$ and $t_2$ are auditing parameters.
The choosing strategy of $t_1$ and $t_2$ is provided in Append \ref{app:ceremony}.

\subsection{End-to-end System}

We now describe the end-to-end execution flow of \name, showing how randomized auditing, logit commitments, and verifiable computation are integrated into a practical LLM inference service.



\paragraph{Setup: model commitment.}
At deployment time, the model provider commits to the claimed full-precision reference model
$\mathcal{M}_{\theta^{\star}}$ by publishing a cryptographic hash commitment.
This commitment binds the server to a specific reference model while preserving full model confidentiality.

\paragraph{Online inference and logging.}
For each incoming query $\vec{x}$, the server performs standard LLM inference.
During execution, the server additionally records the logits $\{\ell_i\}$ at each discrete decision point
(e.g., token selection or expert routing).
Along with the response $(\vec{y}, T)$, $\mathsf{Srv}$ also returns a commitment to $\{\ell_i\}$ to the user.

We estimate logging cost of commercial open-source LLMs
Even the largest open-source models require caching only about 1\,KB of data per token, so this logging and commitment step incurs negligible overhead.
The cached messages can be deleted after a prescribed threshold if the user does not request auditing,
and therefore this mechanism does not introduce significant storage costs.

\paragraph{Randomized auditing.}
$\mathsf{Adt}$ submits random queries to $\mathsf{Srv}$ while behaving indistinguishably from a normal user.
After receiving responses, $\mathsf{Adt}$ \emph{immediately} reveals its auditor identity and requests an audit proof from $\mathsf{Srv}$. 
In response, $\mathsf{Srv}$ invokes a verifiable computation (VC) procedure to compute and prove the corresponding LDD. 
These LDDs serve as observable auditing footprints, providing statistical evidence of server-side
execution integrity.

Formal pseudo-code description is given in Appendix \ref{app:alg}.
In Appendix \ref{app:adaptive}, we also analyze a strategic attacks from $\mathsf{Srv}$, and present the best strategy for a cost-incentive server is honestly commit runtime logits of its deployed model.

\section{Experiments} \label{sec:exp}

We evaluate \name along two complementary dimensions:
(1) its completeness and detection effectiveness, and
(2) the system's imposed overhead. Concretely, our experiments answer the following questions:
\begin{enumerate}
    \item Does the proposed logit distance distribution (LDD) meaningfully characterize approximation fidelity between a deployed execution and the claimed full-precision model?
    \item Can LDD reliably distinguish benign executions from malicious behaviors such as aggressive quantization and model substitution?
    \item What per-request detection probability and false positive rate can be achieved using LDD-based decision rules?
    \item When combined with randomized auditing, does \name provide strong completeness and soundness guarantees against an $\alpha$-dishonest server?
    \item What is the end-to-end performance overhead introduced by \name for model providers?
\end{enumerate}

\subsection{Experimental Setup}

\paragraph{System implementation.}
We build our inference system atop \texttt{vLLM}~\cite{kwon2023efficient}, a widely used high-performance LLM inference framework.
Verifiable computation (VC) is implemented using HuggingFace Transformers~\cite{wolf2019huggingface} and executed in FP32 arithmetic within a TDX enclave, serving as the reference full-precision execution.
For all experiments, we adopt Top-20 token sampling strategy.

\paragraph{Hardware.} 
All experiments are conducted on NVIDIA RTX 6000 Pro GPUs with 96 GB of GPU memory.
We perform inference using tensor parallelism with a degree of 2 across GPUs.

\paragraph{Datasets.}
We evaluate on GSM8K, TriviaQA, and WebQuestions, representing mathematical reasoning, factual QA, and open-domain QA, respectively.
For scalability, we subsample 500 prompts per dataset, resulting in a total of 1,500 evaluation queries.
An initial setup phase is conducted using the first 200 examples from each of GSM8K, TriviaQA, and WebQuestions.
We randomly sample 200 prompts from each dataset to form an independent setup set for calibrating auditing thresholds; all remaining prompts are held out and used exclusively for evaluation.

\paragraph{Models.}
We evaluate both dense and mixture-of-experts (MoE) architectures, including dense models LLaMA3-70B, Qwen3-32B and MoE models Qwen3-30B-A3B, DeepSeek-V2-Lite.

\paragraph{Attacks.} 
We evaluate \name against model substitution and aggressive quantization.
We omit token overreporting from separate empirical evaluation because we have proved in Section \ref{sec:hybrid-compute}, it can be reduced to a special case of model substitution.

For model substitution attacks, we evaluates 1) Substitute LLaMA3-70B with LLaMA3-8B; 2) Substitute Qwen3-32B with Qwen3-14B.
We do not explore model substitution attacks for MoE models because larger versions are too costly to run.
For quantization attacks, we assume all benign deployments are under BF16, and consider FP8 as quantization attacks, following the setting in \cite{cai2025you}.

\subsection{Global Logit Discrepancy Distribution}

We first evaluate whether the global LDD reflects approximation fidelity under different deployment regimes.
Specifically, we compare three settings against the same full-precision reference model: 1) benign BF16 execution; 2) FP8 quantized execution; 3) model substitution (for dense models only).


\begin{table}[h]
\caption{Proportion of TV distance larger than certain threshold} \label{tab:global_threshold}
\begin{adjustbox}{width=\columnwidth}
\begin{tabular}{|c|c||c|c|c|}
\hline
Model  &  Deploy  & $> 0.1$ & $> 0.2$ & $> 0.3$  \\
\hline \hline
\multirow{3}{*}{LLaMA3} & BF16 & $0.017\%$  &  $1.7\times 10^{-5}$   &   $1.9 \times 10^{-6}$     \\
\cline{2-5}
                       & FP8  &  $1.7\%$  &  $0.23\%$   &  $0.054\%$  \\
\cline{2-5}
                       & Sub. &  $11\%$  &  $8.2\%$   &  $6.2\%$  \\
                       \hline
\multirow{3}{*}{Qwen3} & BF16 &  $0.17\%$  &  $0.025\%$  &  $6.0 \times 10^{-5}$  \\
\cline{2-5}
                       & FP8  &  $4.8\%$  &  $1.1\%$   &  $0.42\%$  \\
\cline{2-5}
                       & Sub. &  $28\%$  &  $20\%$  &  $15\%$  \\        
                       \hline
\multirow{2}{*}{Qwen3MoE} & BF16 &  $0.074\%$  &  $2.6 \times 10^{-5}$   &  $4.7 \times 10^{-6}$   \\
\cline{2-5}
                       & FP8  &  $5.5\%$   &   $1.0\%$    &  $0.23\%$  \\
                       \hline
\multirow{2}{*}{Deepseek} & BF16 & $6.6 \times 10^{-5}$   &  $6.0 \times 10^{-6}$    &  $8.5 \times 10^{-7}$ \\
\cline{2-5}
                       & FP8  &  $1.6\%$  &  $0.19\%$   &  $0.053\%$   \\
                       \hline
\end{tabular}
\end{adjustbox}
\end{table}

Figure~\ref{fig:global_tv} has visualized the aggregated logit TV-distance distributions across all evaluated prompts.
Across all evaluated dense and MoE models, benign BF16 execution induces sharply concentrated distributions with rapidly decaying tails. 
In contrast, FP8 quantization substantially increases the tail mass at moderate and large logit distances, while model substitution causes a pronounced right-shift, yielding orders-of-magnitude higher probability mass in the extreme tail.

To quantify the separation, we further report in Table~\ref{tab:global_threshold} the probability that the total-variation (TV) distance between deployed and reference logits exceeds different thresholds. 
It shows that FP8 quantization increases these tail probabilities by several orders of magnitude, while model substitution amplifies them further.
These results confirm that tail events in the LDD provide a highly discriminative signal for distinguishing benign execution from economically motivated deviations, motivating the threshold-based detection rule used in subsequent experiments.

In Appendix \ref{app:global-ldd}, we also present global LDDs beyond TV distances.
All results show distinct tail-probability across different deployments.

\subsection{Per-request Decision Based on LDD}

\begin{table}[h]
\caption{Malicious request detection rate (the rate of malicious requests that can be detected under $\mathrm{FP} = 10^{-5}$). 
As shown in Section \ref{sec:random}, 1\% per-request detection probability is sufficient for effective randomized auditing. } \label{tab:detection}
\begin{adjustbox}{width=\columnwidth}
\begin{tabular}{|c|c||c|c|c|}
\hline
Model  &  Deploy  & GSM8K & TriviaQA & WebQuestions  \\
\hline \hline
\multirow{2}{*}{LLaMA3} & FP8  &  $9.0\%$  &  $9.7\%$   &  $5.3\%$  \\
\cline{2-5}
                       & Sub. &  $42\%$  &  $95\%$   &  $99\%$  \\
                       \hline
\multirow{2}{*}{Qwen3} & FP8  &  $1.3\%$  &  $3.6\%$   &  $2.0\%$  \\
\cline{2-5}
                       & Sub. &  $96\%$  &  $96\%$  &  $97\%$  \\        
                       \hline
Qwen3MoE &  FP8  &  $2.0\%$   &   $2.1\%$    &  $5.0\%$  \\
                       \hline
Deepseek &  FP8  &  $10.3\%$  &  $3.0\%$   &  $3.7\%$   \\
                       \hline
\end{tabular}
\end{adjustbox}
\end{table}

While global LDDs exhibit clear separation, auditing decisions must ultimately be made at the per-request level. In practice, a malicious LLM provider may apply heterogeneous deployment strategies across requests, thereby blurring the global distribution and reducing its discriminative power.

Following Section \ref{sec:ldd-use}, we define a per-request statistic $p(t_1)$ based on the fraction of logit distance exceeding a threshold $t_1$.
The request is flagged as malicious if $p(t_1) > t_2$.
The auditing parameters $t_1, t_2$ is computed via method given in Section \ref{sec:ldd-use}.


Table \ref{tab:detection} present the resulting detection rates.
For quantization attacks, a request is detected with probability at least $1.3\%$, whereas model substitution is detected with probability larger than $40\%$.

Accurately estimating the per-request false positive rate under benign execution is challenging due to the extreme rarity of tail events. 
Owing to computational constraints, our experiments evaluate only on thousands of requests; consequently, all benign requests observed in practice are correctly classified, making direct empirical estimation of the false positive rate infeasible.

To address this limitation, we adopt an extreme value theory (EVT)–based approach to model the tail behavior of benign LDDs. 
Specifically, we fit an EVT model to the extreme tail of the benign distribution $p(t_1)$ and estimate the probability that a benign request produces an LDD exceeding $t_2$. 
Results show that among all datasets and models, the false positive rate is smaller than $10^{-5}$.

\paragraph{Completeness and soundness.}
Assume \name audits approximately 3,000 requests per day.
To achieve overwhelming completeness, among 3,000 requests, $\mathsf{Adt}$ allows at most 3 of them classified to be malicious. 
Consider the at most $10^{-5}$ per-request false positive rate, the probability of rejecting a benign server is
\[
1 - \sum_{k=0}^{3} \binom{3000}{k} (10^{-5})^k (1 - 10^{-5})^{3000-k} \le 10^{-7},
\]
representing overwhelming completeness.

For an $\alpha=0.1$ dishonest server, each audited request is detected with probability at least $10^{-3}$.
The probability of observing 4 or more detections in a day is therefore
\[
1 - \sum_{k=0}^{3} \binom{3000}{k} (10^{-3})^k (1-10^{-3})^{3000-k} \ge 0.3.
\]
Over a month-long horizon, the probability of persistent evasion becomes negligible.

\subsection{Auditing Parameter Study}

\begin{table}[h]
\caption{Detection rate and false positive rate under different hyperparameter} \label{tab:ablation}
\begin{adjustbox}{width=\columnwidth}
\begin{tabular}{|c|c||c|c|c|c|}
\hline
Hyper. &  & LLaMA3 & Qwen3 & Qwen3MoE &  Deepseek \\
\hline \hline
\multirow{2}{*}{0.02}  &  Detect &  $2.1\%$  &  $0.44\%$   &  $0.93\%$  &  $1.3\%$ \\
\cline{2-6}
 &  FP &  $2 \times 10^{-9}$  & $2 \times 10^{-8}$  &  $5 \times 10^{-10}$  &  $8 \times 10^{-10}$ \\
\hline 
\multirow{2}{*}{0.03} &  Detect  & $6.2\%$  & $0.67\%$  &  $2.1\%$   &  $3.0\%$  \\
\cline{2-6}
&  FP &  $1 \times 10^{-8}$  &  $3 \times 10^{-7}$   &  $1 \times 10^{-8}$  &  $2 \times 10^{-8}$ \\
\hline
\multirow{2}{*}{0.04} & Detect & $6.2\%$  &  $1.2\%$    &   $2.1\%$  &  $3.0\%$  \\
\cline{2-6}
&  FP &  $1 \times 10^{-8}$  &  $1 \times 10^{-6}$    &  $1 \times 10^{-8}$  &  $2 \times 10^{-8}$ \\
\hline
\multirow{2}{*}{0.05} & Detect & $6.2\%$ &  $2.3\%$  &  $3.6\%$  &  $5.6\%$ \\    
\cline{2-6}
&  FP &  $1 \times 10^{-8}$ &  $4 \times 10^{-6}$   &  $9 \times 10^{-7}$  &  $8 \times 10^{-8}$ \\    
\hline
\multirow{2}{*}{0.06} & Detect & $9.4\%$   &   $3.3\%$  &  $3.6\%$  &  $8.8\%$  \\
\cline{2-6}
&  FP &  $4 \times 10^{-7}$  &  $2 \times 10^{-5}$   &  $9 \times 10^{-7}$  &  $5 \times 10^{-7}$ \\    
                       \hline
\multirow{2}{*}{0.07} & Detect &  $9.4\%$  &  $5.1\%$   &  $5.2\%$ & $8.8\%$ \\
\cline{2-6}
&  FP &  $4 \times 10^{-7}$  &  $6 \times 10^{-5}$    &  $4 \times 10^{-6}$  &  $5 \times 10^{-7}$ \\    
                       \hline
\end{tabular}
\end{adjustbox}
\end{table}

The derivation of the auditing parameters is largely heuristic, so we explore a wider range of parameter choices to demonstrate the robustness of our scheme.  
When computing $t_1$ and $t_2$ via the ceremony, we employ an ad-hoc hyperparameter that accepts only $(t_1, t_2)$ pairs achieving a 5\% detection rate on the ceremony dataset. 
Here, we vary this hyperparameter around the 5\% threshold to derive different auditing parameters $t_1$ and $t_2$, and report their corresponding false positive rates and detection rates in Table~\ref{tab:detection}.  
Due to space constraints, we aggregate the results from three datasets to present the overall false positive rate and detection rate.
Results show that for all derived auditing parameters, the FP rate and detection rate are always acceptable.

\subsection{Overhead to Benign Execution}

We finally evaluate the efficiency impact of \name.
Since verifiable computation is triggered only on a small audited subset, the dominant overhead arises from storing and committing lightweight runtime values.

Across all evaluated models, the end-to-end throughput overhead for a benign server is below $1\%$, demonstrating practical deployability.
The VC deployed in CPU TEE is hundreds of times slower than inference.
Consider that the proportion of audited request is smaller than $10^{-5}$, the overall cost is negligible.

\begin{table}[h]
    \caption{Overhead Evaluation. `Thp.' denotes the inference throughput loss. `Request VC' denotes the runtime of VC on CPU-TEE compared with GPU inference}
    \centering
    \begin{tabular}{|c|c|c|}
    \hline
        & Thp. & Request VC \\
    \hline
       LLaMA3-70B & $0.3\%$ &  $400 \times$  \\
    \hline 
       Qwen3-32B & $0.3\%$ &  $400 \times$  \\
    \hline
      Qwen3-30B-A3B & $0.9\%$ &  $900 \times$  \\
    \hline
      DeepSeek-V2-Lite & $1.0\%$ &  $800 \times$  \\
    \hline
    \end{tabular}
\end{table}



\section{Conclusion}

We have introduced \name, a practical and robust auditing framework for LLM inference services that operate as black-box APIs. 
By combining selective verifiable computation with a novel Logit Distance Distribution (LDD) metric, \name enables auditors to detect economically motivated deviations—such as model substitution, aggressive quantization, and token overreporting—without requiring access to model internals or relying on trusted hardware.

Looking ahead, we believe that \name lays the foundation for a new class of auditing frameworks. Future work may explore additional optimizations, particularly in leveraging LDD to more effectively identify malicious requests. 
Ultimately, our framework represents a step toward greater transparency, accountability, and trust in commercial LLM services.


\bibliography{BIB}

@article{sun2025coin,
  title={CoIn: Counting the Invisible Reasoning Tokens in Commercial Opaque LLM APIs},
  author={Sun, Guoheng and Wang, Ziyao and Tian, Bowei and Liu, Meng and Shen, Zheyu and He, Shwai and He, Yexiao and Ye, Wanghao and Wang, Yiting and Li, Ang},
  journal={arXiv preprint arXiv:2505.13778},
  year={2025}
}

@inproceedings{qu2025zkgpt,
  title={zkGPT: An Efficient Non-interactive Zero-knowledge Proof Framework for LLM Inference},
  author={Qu, Wenjie and Sun, Yijun and Liu, Xuanming and Lu, Tao and Guo, Yanpei and Chen, Kai and Zhang, Jiaheng},
  booktitle={34st USENIX Security Symposium (USENIX Security 25)},
  year={2025}
}

@inproceedings{sun2024zkllm,
  title={zkllm: Zero knowledge proofs for large language models},
  author={Sun, Haochen and Li, Jason and Zhang, Hongyang},
  booktitle={Proceedings of the 2024 on ACM SIGSAC Conference on Computer and Communications Security},
  pages={4405--4419},
  year={2024}
}

@article{brown2020language,
  title={Language models are few-shot learners},
  author={Brown, Tom and Mann, Benjamin and Ryder, Nick and Subbiah, Melanie and Kaplan, Jared D and Dhariwal, Prafulla and Neelakantan, Arvind and Shyam, Pranav and Sastry, Girish and Askell, Amanda and others},
  journal={Advances in neural information processing systems},
  volume={33},
  pages={1877--1901},
  year={2020}
}

@article{chowdhery2023palm,
  title={Palm: Scaling language modeling with pathways},
  author={Chowdhery, Aakanksha and Narang, Sharan and Devlin, Jacob and Bosma, Maarten and Mishra, Gaurav and Roberts, Adam and Barham, Paul and Chung, Hyung Won and Sutton, Charles and Gehrmann, Sebastian and others},
  journal={Journal of Machine Learning Research},
  volume={24},
  number={240},
  pages={1--113},
  year={2023}
}

@inproceedings{adamopoulou2020overview,
  title={An overview of chatbot technology},
  author={Adamopoulou, Eleni and Moussiades, Lefteris},
  booktitle={IFIP international conference on artificial intelligence applications and innovations},
  pages={373--383},
  year={2020},
  organization={Springer}
}

@article{chen2021evaluating,
  title={Evaluating large language models trained on code},
  author={Chen, Mark},
  journal={arXiv preprint arXiv:2107.03374},
  year={2021}
}

@article{wei2022chain,
  title={Chain-of-thought prompting elicits reasoning in large language models},
  author={Wei, Jason and Wang, Xuezhi and Schuurmans, Dale and Bosma, Maarten and Xia, Fei and Chi, Ed and Le, Quoc V and Zhou, Denny and others},
  journal={Advances in neural information processing systems},
  volume={35},
  pages={24824--24837},
  year={2022}
}

@misc{openai_api,
  title = {OpenAI API Documentation},
  author = {{OpenAI}},
  year = {2023},
  howpublished = {\url{https://platform.openai.com/docs}}
}

@misc{google_vertex_ai,
  title = {Vertex AI Generative AI},
  author = {{Google Cloud}},
  year = {2023},
  howpublished = {\url{https://cloud.google.com/vertex-ai}}
}

@misc{anthropic_api,
  title = {Anthropic API Documentation},
  author = {{Anthropic}},
  year = {2023},
  howpublished = {\url{https://docs.anthropic.com}}
}

@misc{gpt_lazy1,
  title={GPT4-Turbo more “stupid/lazy” - It’s not a GPT4}, 
  url={https://community.openai.com/t/gpt4-turbo-more-stupid-lazy-its-not-a-gpt4/608008},
  author={webtailken},
  year={2024},
}

@misc{gpt_lazy2,
  title={OpenAI did made GPT3.5 more stupid?}, 
  url={https://community.openai.com/t/openai-did-made-gpt3-5-more-stupid/262979},
  author={theevildays},
  year={2024},
}

@article{chen2023frugalgpt,
  title={Frugalgpt: How to use large language models while reducing cost and improving performance},
  author={Chen, Lingjiao and Zaharia, Matei and Zou, James},
  journal={arXiv preprint arXiv:2305.05176},
  year={2023}
}

@inproceedings{mohammadi2025evaluation,
  title={Evaluation and benchmarking of llm agents: A survey},
  author={Mohammadi, Mahmoud and Li, Yipeng and Lo, Jane and Yip, Wendy},
  booktitle={Proceedings of the 31st ACM SIGKDD Conference on Knowledge Discovery and Data Mining V. 2},
  pages={6129--6139},
  year={2025}
}

@article{cai2025you,
  title={Are you getting what you pay for? auditing model substitution in llm apis},
  author={Cai, Will and Shi, Tianneng and Zhao, Xuandong and Song, Dawn},
  journal={arXiv preprint arXiv:2504.04715},
  year={2025}
}

@inproceedings{tan2025pipellm,
  title={Pipellm: Fast and confidential large language model services with speculative pipelined encryption},
  author={Tan, Yifan and Tan, Cheng and Mi, Zeyu and Chen, Haibo},
  booktitle={Proceedings of the 30th ACM International Conference on Architectural Support for Programming Languages and Operating Systems, Volume 1},
  pages={843--857},
  year={2025}
}

@inproceedings{jouppi2017datacenter,
  title={In-datacenter performance analysis of a tensor processing unit},
  author={Jouppi, Norman P and Young, Cliff and Patil, Nishant and Patterson, David and Agrawal, Gaurav and Bajwa, Raminder and Bates, Sarah and Bhatia, Suresh and Boden, Nan and Borchers, Al and others},
  booktitle={Proceedings of the 44th annual international symposium on computer architecture},
  pages={1--12},
  year={2017}
}

@misc{aws_inferentia,
  title={AWS Inferentia}, 
  url={https://aws.amazon.com/ai/machine-learning/inferentia/},
  author={AWS},
  year={2022},
}

@inproceedings{zhou2025accelerating,
  title={Accelerating Model Training on Ascend Chips: An Industrial System for Profiling, Analysis and Optimization},
  author={Zhou, Yuhang and Wang, Zibo and Wang, Zhibin and Zhang, Ruyi and Tian, Chen and Wang, Xiaoliang and Dou, Wanchun and Chen, Guihai and Wang, Bingqiang and Tian, Yonghong and others},
  booktitle={2025 USENIX Annual Technical Conference (USENIX ATC 25)},
  pages={1387--1408},
  year={2025}
}

@inproceedings{chen2024zkml,
  title={Zkml: An optimizing system for ml inference in zero-knowledge proofs},
  author={Chen, Bing-Jyue and Waiwitlikhit, Suppakit and Stoica, Ion and Kang, Daniel},
  booktitle={Proceedings of the Nineteenth European Conference on Computer Systems},
  pages={560--574},
  year={2024}
}

@article{gao2024model,
  title={Model Equality Testing: Which Model Is This API Serving?},
  author={Gao, Irena and Liang, Percy and Guestrin, Carlos},
  journal={arXiv preprint arXiv:2410.20247},
  year={2024}
}

@incollection{merkle2019protocols,
  title={Protocols for public key cryptosystems},
  author={Merkle, Ralph C},
  booktitle={Secure communications and asymmetric cryptosystems},
  pages={73--104},
  year={2019},
  publisher={Routledge}
}

@incollection{goldwasser2019knowledge,
  title={The knowledge complexity of interactive proof-systems},
  author={Goldwasser, Shafi and Micali, Silvio and Rackoff, Chales},
  booktitle={Providing sound foundations for cryptography: On the work of shafi goldwasser and silvio micali},
  pages={203--225},
  year={2019}
}

@inproceedings{sabt2015trusted,
  title={Trusted execution environment: What it is, and what it is not},
  author={Sabt, Mohamed and Achemlal, Mohammed and Bouabdallah, Abdelmadjid},
  booktitle={2015 IEEE Trustcom/BigDataSE/Ispa},
  volume={1},
  pages={57--64},
  year={2015},
  organization={IEEE}
}

@inproceedings{gennaro2013quadratic,
  title={Quadratic span programs and succinct NIZKs without PCPs},
  author={Gennaro, Rosario and Gentry, Craig and Parno, Bryan and Raykova, Mariana},
  booktitle={Annual International Conference on the Theory and Applications of Cryptographic Techniques},
  pages={626--645},
  year={2013},
  organization={Springer}
}

@inproceedings{kwon2023efficient,
  title={Efficient memory management for large language model serving with pagedattention},
  author={Kwon, Woosuk and Li, Zhuohan and Zhuang, Siyuan and Sheng, Ying and Zheng, Lianmin and Yu, Cody Hao and Gonzalez, Joseph and Zhang, Hao and Stoica, Ion},
  booktitle={Proceedings of the 29th symposium on operating systems principles},
  pages={611--626},
  year={2023}
}

@article{wolf2019huggingface,
  title={Huggingface's transformers: State-of-the-art natural language processing},
  author={Wolf, Thomas and Debut, Lysandre and Sanh, Victor and Chaumond, Julien and Delangue, Clement and Moi, Anthony and Cistac, Pierric and Rault, Tim and Louf, R{\'e}mi and Funtowicz, Morgan and others},
  journal={arXiv preprint arXiv:1910.03771},
  year={2019}
}

@article{chen2024chatgpt,
  title={How is ChatGPT’s behavior changing over time?},
  author={Chen, Lingjiao and Zaharia, Matei and Zou, James},
  journal={Harvard Data Science Review},
  volume={6},
  number={2},
  year={2024},
  publisher={The MIT Press}
}

@inproceedings{eyuboglu2024model,
  title={Model changelists: Characterizing updates to ml models},
  author={Eyuboglu, Sabri and Goel, Karan and Desai, Arjun and Chen, Lingjiao and Monfort, Mathew and R{\'e}, Chris and Zou, James},
  booktitle={Proceedings of the 2024 ACM Conference on Fairness, Accountability, and Transparency},
  pages={2432--2453},
  year={2024}
}

@article{sun2025idiosyncrasies,
  title={Idiosyncrasies in large language models},
  author={Sun, Mingjie and Yin, Yida and Xu, Zhiqiu and Kolter, J Zico and Liu, Zhuang},
  journal={arXiv preprint arXiv:2502.12150},
  year={2025}
}

@article{ong2025toploc,
  title={Toploc: A locality sensitive hashing scheme for trustless verifiable inference},
  author={Ong, Jack Min and Di Ferrante, Matthew and Pazdera, Aaron and Garner, Ryan and Jaghouar, Sami and Basra, Manveer and Ryabinin, Max and Hagemann, Johannes},
  journal={arXiv preprint arXiv:2501.16007},
  year={2025}
}

@techreport{nvidia2023hoppercc,
  author      = {{NVIDIA}},
  title       = {Confidential Compute on NVIDIA Hopper H100},
  institution = {NVIDIA},
  year        = {2023},
}

@article{huang2025exploring,
  title={Exploring and mitigating adversarial manipulation of voting-based leaderboards},
  author={Huang, Yangsibo and Nasr, Milad and Angelopoulos, Anastasios and Carlini, Nicholas and Chiang, Wei-Lin and Choquette-Choo, Christopher A and Ippolito, Daphne and Jagielski, Matthew and Lee, Katherine and Liu, Ken Ziyu and others},
  journal={arXiv preprint arXiv:2501.07493},
  year={2025}
}

@article{yao2025nondeterminism,
  title={Nondeterminism-Aware Optimistic Verification for Floating-Point Neural Networks},
  author={Yao, Jianzhu and Su, Hongxu and Liao, Taobo and Cheng, Zerui and Zhang, Huan and Wang, Xuechao and Viswanath, Pramod},
  journal={arXiv preprint arXiv:2510.16028},
  year={2025}
}

@article{zhu2025auditing,
  title={Auditing Black-Box LLM APIs with a Rank-Based Uniformity Test},
  author={Zhu, Xiaoyuan and Ye, Yaowen and Qiu, Tianyi and Zhu, Hanlin and Tan, Sijun and Mannan, Ajraf and Michala, Jonathan and Popa, Raluca Ada and Neiswanger, Willie},
  journal={arXiv preprint arXiv:2506.06975},
  year={2025}
}

@article{velasco2025your,
  title={Is your LLM overcharging you? tokenization, transparency, and incentives},
  author={Velasco, Ander Artola and Tsirtsis, Stratis and Okati, Nastaran and Gomez-Rodriguez, Manuel},
  journal={arXiv preprint arXiv:2505.21627},
  year={2025}
}
\bibliographystyle{icml2026}

\newpage
\appendix

\section{Extended Related Work} \label{app:related}

Existing LLM auditing frameworks struggle to simultaneously achieve \emph{economy}, \emph{proprietorship}, and \emph{robustness}.  
Section~\ref{sec:related} discusses the limitations of cryptography-based approaches and GPU TEEs in terms of deployment cost and infrastructure requirements.  
In this section, we review additional auditing paradigms and analyze their limitations under the black-box service setting.

\paragraph{Direct Verification.}
When the model is fully public, auditing becomes straightforward.  
The most direct approach is to re-execute inference and compare the outputs with those returned by the service.  
TopLoc~\cite{ong2025toploc} accelerates verification by transforming decoding operations into prefilling, reducing the computational overhead of re-execution.  
Yao et al.~\cite{yao2025nondeterminism} propose verifying only a subset of randomly selected layers and integrate the auditing process with a blockchain-based logging mechanism.  
However, these methods require the auditor to access model weights and execution details, making them unsuitable for proprietary commercial models.

\paragraph{Detecting LLM-generated text.}
Another line of work attempts to audit black-box services by analyzing their output behavior.  
Prior studies track performance drift and behavioral changes over time~\cite{chen2024chatgpt,eyuboglu2024model}.  
Other methods aim to identify the deployed model by examining output distributions, including classifier-based fingerprinting~\cite{sun2025idiosyncrasies}, statistical discrepancy tests~\cite{gao2024model,zhu2025auditing}, and identity-style prompting~\cite{huang2025exploring}.  
However, these approaches typically require large query budgets and are vulnerable to randomized model substitution or decoding variability.  
Moreover, since they rely solely on observable text outputs, they cannot detect attacks such as token overreporting, which do not affect the returned content.

\paragraph{Token Overreport Detection.}
CoIn~\cite{sun2025coin} addresses the token overreporting issue by committing embeddings of hidden reasoning tokens in a Merkle hash tree for quantity verification, and by performing embedding-based relevance checks to detect low-effort or fabricated content.  
Nevertheless, a provider may still generate hidden tokens using a cheaper or simplified model while preserving semantic consistency, allowing the inflated sequence to pass such semantic validation.  
Complementary work by Velasco et al.~\cite{velasco2025your} analyzes the problem from an economic perspective, showing that the pay-per-token pricing mechanism inherently incentivizes overreporting under information asymmetry. They argue that, since provider revenue scales with reported token length while users cannot observe the internal generation process, strategic tokenization or count manipulation becomes financially attractive, and propose pay-per-character pricing as an incentive-compatible alternative.  
However, such pricing changes do not provide technical guarantees on faithful execution.

\paragraph{Summary.}
Existing approaches each address only part of the auditing challenge.  
Direct re-execution methods provide strong guarantees but require full model access, violating proprietorship.  
Output-based statistical methods are economical but lack robustness against adaptive or invisible deviations.  
Economic or billing-based solutions address incentive misalignment but do not verify execution correctness.  
In contrast, \name aims to achieve a balanced design that simultaneously ensures economy, preserves model confidentiality, and provides provable detection guarantees for execution deviations under black-box access.
\section{Algorithms} \label{app:alg}

The pseudo-code is presented at Algorithm \ref{alg:srv}, Algorithm \ref{alg:adt} and Algorithm \ref{alg:vc}.

\begin{algorithm}
\caption{$\mathsf{Srv}$: LLM Server's algorithm} \label{alg:srv}
\begin{algorithmic}[1]
    \REQUIRE LLM model $\mathcal{M}_{\theta} := (E_{\theta}, F_{\theta}, S, G_{\theta}, D)$
    \STATE $\psi_M := \mathsf{Hash}(\mathcal{M}_{\theta^\star})$, where $\mathcal{M}_{\theta^{\star}}$ is the full-precision version of $\mathcal{M}_{\theta}$.
    \STATE Publish $\psi_M$
    \WHILE{receiving $\angles{\mathsf{Request}, \vec{x}}$ from $\mathsf{Usr}$}
    \STATE $r \leftarrow_{\$} \braces{0, 1}^{\lambda}$
    \STATE $h_0 \leftarrow E_{\theta}(\vec{x})$.
    \FOR{$i = 1, 2, \cdots, N$}
    \STATE $\tilde{h}_i, \ell_i \leftarrow F_{\theta}(h_{i-1})$
    \STATE $d_i := S(\ell_i, r)$
    \STATE $h_i \leftarrow G_{\theta}(\tilde{h}_i, d_i)$
    \ENDFOR
    \STATE $\vec{y}, T := D(\braces{d_1, \cdots, d_N})$
    \STATE $\psi := \mathsf{Hash}(r, \braces{\ell_1, \cdots, \ell_N})$
    \STATE \textbf{send} $\angles{\mathsf{Response}, \vec{y}, T, \psi}$ to $\mathsf{Usr}$.

    \WHILE{receiving $\angles{\mathsf{Audit}}$ from $\mathsf{Adt}$}
    \STATE $w := \parens{\mathcal{M}_{\theta^{\star}}, r, \braces{\ell_1, \cdots, \ell_N}}$
    \STATE $\phi, \pi := \mathsf{VC}.\mathsf{prove}\parens{(\psi_M, \psi, \vec{x}, \vec{y}, T); w}$
    \STATE \textbf{send} $\angles{\mathsf{Proof}, \phi, \pi}$ to $\mathsf{Adt}$.
    \ENDWHILE
    \ENDWHILE
\end{algorithmic}
\end{algorithm}

\begin{algorithm}
\caption{$\mathsf{Adt}$: LLM Auditor's algorithm} \label{alg:adt}
\begin{algorithmic}[1]
    \REQUIRE Auditing parameter $\mathsf{ap}$
    \ENSURE $b \in \braces{0, 1, \bot}$
    \STATE Generate random prompt $\vec{x}$
    \STATE \textbf{send} $\angles{\mathsf{Request}, \vec{x}}$ to $\mathsf{Srv}$
    \WHILE{receiving $\angles{\mathsf{Response}, \vec{y}, T, \psi}$}
    \STATE \textbf{send} $\angles{\mathsf{Audit}}$ to $\mathsf{Srv}$.
    \ENDWHILE

    \WHILE{receiving $\angles{\mathsf{Proof}, \phi, \pi}$}
    \IF{$\mathsf{VC}.\mathsf{prove}\parens{(\psi_M, \psi, \vec{x}, \vec{y}, T), \pi}$}
    \STATE Return $\bot$.
    \ENDIF
    \STATE Return 0/1 based on $\phi$ and $\mathsf{ap}$.
    \ENDWHILE
\end{algorithmic}
\end{algorithm}

\begin{algorithm}
\caption{$\mathsf{VC}$: Verifiable computation function} \label{alg:vc}
\begin{algorithmic}[1]
    \REQUIRE Public input: model commitment $\psi_M$, logits commitment $\psi$, prompt $\vec{x}$, response $\vec{y}$, token usage $T$;
    Private input: model $\mathcal{M}_{\theta^\star} = (E_{\theta^\star}, F_{\theta^\star}, S, G_{\theta^\star}, D)$, seed $r$, logits $\braces{\ell_1, \cdots, \ell_N}$.
    \ENSURE Logits distance distribution
    \IF{$\psi \neq \mathsf{Hash}\parens{r, \braces{\ell_1, \cdots, \ell_N}}$ \textbf{or} $\mathsf{\psi}_M \neq \mathsf{Hash}(\mathcal{M}_{\theta^\star})$}
    \STATE \textbf{abort}
    \ENDIF
    \STATE $h^{\star}_0 \leftarrow E_{\theta^\star}(\vec{x})$.
    \FOR{$i = 1, 2, \cdots, N$} 
    \STATE $\tilde{h}^{\star}_i, \ell^{\star}_i \leftarrow F_{\theta^\star}(h^{\star}_{i-1})$
    \STATE $\delta_i = \mathsf{Dis}(\ell_i, \ell^{\star}_i)$
    \STATE $d_i := S(\ell_i, r)$
    \STATE $h^{\star}_i \leftarrow G_{\theta^\star}(\tilde{h}^{\star}_i, d_i)$
    \ENDFOR
    \IF{$\vec{y}, T \neq D(\braces{d_1, \cdots, d_N})$}
    \STATE \textbf{abort}    
    \ENDIF
    \STATE \textbf{return} distribution of $\braces{\delta_1, \cdots, \delta_N}$
\end{algorithmic}
\end{algorithm}



\section{Proof of Propositions} \label{app:proof}

\subsection{Proposition \ref{prop:logit_bias_variance} and Proposition \ref{prop:logit_bias_variance_numerical}} \label{app:logit-bias}

\begin{proof}
Fix a prompt $\vec{x}$ and a sequence of discrete decisions $\vec{d} = (d_1,\dots,d_N)$, thereby aligning the control flow of two executions. Let $\ell_i^{\star}(\theta)$ denote the deterministic logits produced by the full-precision model at step $i$.

For an approximate execution $\mathcal{M}^{(A)}$, the observed logits can be written as
\[
    \ell_i^{(A)} = \ell_i^{\star}(\theta^{(A)}) + \eta_i^{(A)},
\]
where $\eta_i^{(A)}$ captures run-dependent numerical error due to finite precision or kernel non-determinism. Comparing $\mathcal{M}^{(A)}$ with the ideal model $\mathcal{M}^{(B)}=\mathcal{M}^{\star}$ yields
\[
    \ell_i^{(A)} - \ell_i^{(B)} 
    = 
    \underbrace{\big(\ell_i^{\star}(\theta^{(A)}) - \ell_i^{\star}(\theta^{(B)})\big)}_{\Delta_i^{\mathrm{model}}}
    +
    \underbrace{\eta_i^{(A)}}_{\Delta_i^{\mathrm{prec}}},
\]
where $\eta_i^{(B)} = 0$ by definition.

\textit{Model substitution.}  
If $\theta^{(A)} \neq \theta^{(B)}$, then $\Delta_i^{\mathrm{model}} \neq 0$, producing a systematic (biased) deviation in logits.

\textit{Precision reduction.}  
If $\theta^{(A)} = \theta^{(B)}$, then $\Delta_i^{\mathrm{model}} = 0$ and discrepancies arise solely from numerical noise. Lower-precision formats induce higher variance:
\[
    \mathrm{Var}(\eta_i^{\mathrm{FP8}}) > \mathrm{Var}(\eta_i^{\mathrm{FP16}}),
\]
resulting in stochastic, zero-mean deviations rather than persistent bias.
\end{proof}

\subsection{Proposition \ref{prop:token_report}} \label{app:overbill}

\begin{proof}
Let the claimed model be
\[
\mathcal{M} := (E_{\theta}, F_{\theta}, S, G_{\theta}, D),
\]
and consider a fixed prompt $\vec{x}$ with discrete decision sequence
\(
\vec{d} = (d_1,\dots,d_N)
\)
generated by $\mathcal{M}$. The corresponding output is
\[
(\vec{y}, T) = D(\vec{d}), \qquad T = N.
\]

We construct an alternative model
\[
\mathcal{M}' := (E_{\theta}, F'_{\theta}, S, G'_{\theta}, D'),
\]
which differs from $\mathcal{M}$ only in the recurrent hybrid step and output reconstruction. Define $K > 0$ and extend the discrete decision sequence as
\[
\vec{d}' := (d_1,\dots,d_N,\underbrace{\bar{d},\dots,\bar{d}}_{K}),
\]
where $\bar{d} \in \mathcal{D}$ is a fixed dummy decision.

Define the modified state update such that
\[
G'_{\theta}(\tilde{h}, \bar{d}) = h \quad \text{and} \quad
F'_{\theta}(h) = (h, \ell),
\]
i.e., the dummy steps leave the hidden state invariant and do not affect subsequent decisions. Finally, define the reconstruction function
\[
D'(\vec{d}') := (\vec{y}, N+K),
\]
where $\vec{y}$ is recovered solely from the prefix $(d_1,\dots,d_N)$.

By construction,
\[
\mathcal{M}'(\vec{x}) = (\vec{y}, T+K), 
\qquad
\mathcal{M}(\vec{x}) = (\vec{y}, T),
\]
while the observable output sequence $\vec{y}$ is identical. Hence, token overreporting corresponds to replacing $\mathcal{M}$ with a functionally distinct model $\mathcal{M}'$ that modifies the hybrid transition structure.

Since $\mathcal{M}' \neq \mathcal{M}$ but preserves output semantics while altering reported resource usage, token overreporting is equivalent to a form of model substitution, and therefore constitutes a special case of model downsizing.
\end{proof}

\section{Optimization: Verification of Top-$K$.} \label{app:topk-dis}

Among all discrete operations in LLM inference, top-$K$ selection is the most dominant, appearing in both token sampling and mixture-of-experts routing. Formally, top-$K$ is a deterministic operator
\[
\mathsf{Top}_K : \mathbb{R}^n \rightarrow [n]^K,
\]
which returns the indices of the $K$ largest entries in a logits vector. In practice, $n$ is typically much larger than $K$. For example, in token sampling, $n$ corresponds to the size of the entire vocabulary, often exceeding $10^5$, while $K$ is typically around 20.

Due to the large discrepancy between $n$ and $K$, storing and committing to the full logits vector can be prohibitively expensive.
Instead, \name allows the LLM to store only the runtime top-$K$ indices rather than the entire logits vector, by introducing a minimal-perturbation distance criterion.

\paragraph{Top-$K$ distance.}
Given logits $\ell \in \mathbb{R}^n$ and a claimed index set $I \subseteq [n]$ with $|I|=K$, we define
\[
\Delta_{\mathrm{TopK}}(\ell, I)
:=
\min_{\ell' : \mathsf{Top}_K(\ell')=I} \|\ell' - \ell\|_1 .
\]

By construction, $\Delta_{\mathrm{TopK}}(\ell, I) = 0$ if and only if $I$ is a valid top-$K$ set. Otherwise, it quantifies the minimal numerical perturbation required to make $I$ the exact top-$K$. The defined distance is $1$-Lipschitz with respect to $\ell$, ensuring robustness to small numerical noise.

\paragraph{Efficient computation}

For each given logit vector $\ell$ and index set $I$, $\Delta_{\mathrm{TopK}}(\ell, I)$ can be evaluated efficiently. Since $I$ denotes the indices of the largest $K$ entries in the logits vector, there must exist a scalar threshold $t \in \mathbb{R}$ such that
\[
\ell'_i \ge t \quad \forall i \in I,  \qquad \ell'_j \le t \quad \forall j \notin I.
\]
Using a greedy algorithm, for a fixed threshold $t$, the optimal perturbed logits $\ell'$ satisfy
\[
\ell'_i = \max(\ell_i, t) \quad \forall i \in I,
\quad
\ell'_j = \min(\ell_j, t) \quad \forall j \notin I.
\]
Intuitively, entries in $I$ are minimally raised to $t$, while entries outside $I$ are minimally lowered. 

Assume the logits are sorted as $\ell_1 \ge \ell_2 \ge \cdots \ge \ell_n$, and let $t$ lie in the interval $\ell_u \ge t \ge \ell_{u+1}$. The resulting distance is given by
\[
\sum_{i \le u,\, i \notin I} (\ell_i - t) + \sum_{i \ge u + 1,\, i \in I} (t - \ell_i).
\]
When $t \in [\ell_u, \ell_{u+1}]$, this expression is a linear function of $t$, and thus monotonic. The distance reaches its minimum when $t$ is either $\ell_u$ or $\ell_{u+1}$. Therefore, we only need to consider $n$ candidate thresholds, where $t$ equals $\ell_1, \ell_2, \ldots, \ell_n$.
Utilizing this observation, the top-$K$ distance can be computed in linear time.

\paragraph{Integration with \name.}

During inference, the server stores only the top-$K$ indices.
During proving, the VC recomputes logits $\ell$ using the full-precision model and verifies the recorded top-$K$ indices $I$ by evaluating $\Delta_{\mathrm{TopK}}(\ell, I)$.
The resulting distance distribution is then used for statistical verification.
\section{Ceremony: choice of $t_1$, $t_2$} \label{app:ceremony}

To derive the desirable hyperparameters $t_1$ and $t_2$, an initial setup is required: the model provider runs both their deployed LLM and the quantized version they intend to prevent on a prescribed, sufficiently large dataset within a TEE to obtain the empirical per-request LDD as a reference.  
To search for the optimal hyperparameters, the auditor enumerates each $t_1$ and identifies the largest corresponding $t_2$ such that the detection rate exceeds 5\%.  
The $(t_1, t_2)$ pair yielding the lowest false positive rate is selected as the final parameter setting.

\section{Adaptive Adversary Analysis} \label{app:adaptive}

The preceding discussion assumes that the adversary commits to logits generated by a simplified or quantized execution. A more subtle concern is whether a rational adversary could strategically commit to \emph{fabricated} logits to increase the probability of evasion. 
Here, we prove that as long as the LLM server is rational (though not necessarily benign), no such strategy can offer additional benefit.

\paragraph{Rationality assumption.}
We assume the adversary is rational: under a fixed computational budget $T$, it always adopts the highest-quality approximation of the claimed model achievable within that budget.

\begin{proposition}[Logit Commitment Optimality]
\label{prop:unforgeability}
Let $\mathcal{M}_{\mathrm{opt}}$ be the best approximation of the claimed model $\mathcal{M}^\star$ achievable within budget $T$.
Then the adversary’s dominant strategy is to truthfully commit to the logits produced by $\mathcal{M}_{\mathrm{opt}}$.

\end{proposition}

\begin{proof}
Let $\mathcal{M}_T$ denote the set of all inference procedures executable within budget $T$,
and let $L_{\mathcal{M}}(\vec{x})$ be the logits produced by model $\mathcal{M}$ on input $\vec{x}$.
Let $D(\cdot,\cdot)$ be a distance metric measuring deviation from the claimed logits
$L_{\mathcal{M}^{\star}}(\vec{x})$.

By definition,
\[
\mathcal{M}_{\mathrm{opt}}
= \arg\min_{\mathcal{M}\in \mathcal{M}_T} D\!\left(L_{\mathcal{M}}(\vec{x}),\, L_{\mathcal{M}^{\star}}(\vec{x})\right).
\]
Suppose the adversary generates a committed logit vector $L_{\mathrm{fake}}$ such that
\[
D\!\left(L_{\mathrm{fake}},\, L_{\mathcal{M}^{\star}}(\vec{x})\right)
<
D\!\left(L_{\mathcal{M}_{\mathrm{opt}}}(\vec{x}),\, L_{\mathcal{M}^{\star}}(\vec{x})\right).
\]
Producing $L_{\mathrm{fake}}$ requires some computational process within budget $T$,
implying the existence of a model in $\mathcal{M}_T$ that approximates $\mathcal{M}^{\star}$
better than $\mathcal{M}_{\mathrm{opt}}$.
This contradicts the optimality of $\mathcal{M}_{\mathrm{opt}}$.
\end{proof}

\section{Other Global LDDs} \label{app:global-ldd}

We present visualization of global logit TV distance in Figure \ref{fig:global_tv}, logit KL divergence in Figure \ref{fig:global_kl}, token sampling top-$K$ distance in Figure \ref{fig:global_topk_diff} and expert top-$K$ distance in Figure \ref{fig:global_expert_diff}.

\begin{figure}[t]
    \centering
    \begin{subfigure}[t]{0.49\linewidth}
        \centering
        \includegraphics[width=\linewidth]{figures/llama_tvp.pdf}
        \caption{LLaMA3-70B}
    \end{subfigure}
    \hfill
    \begin{subfigure}[t]{0.49\linewidth}
        \centering
        \includegraphics[width=\linewidth]{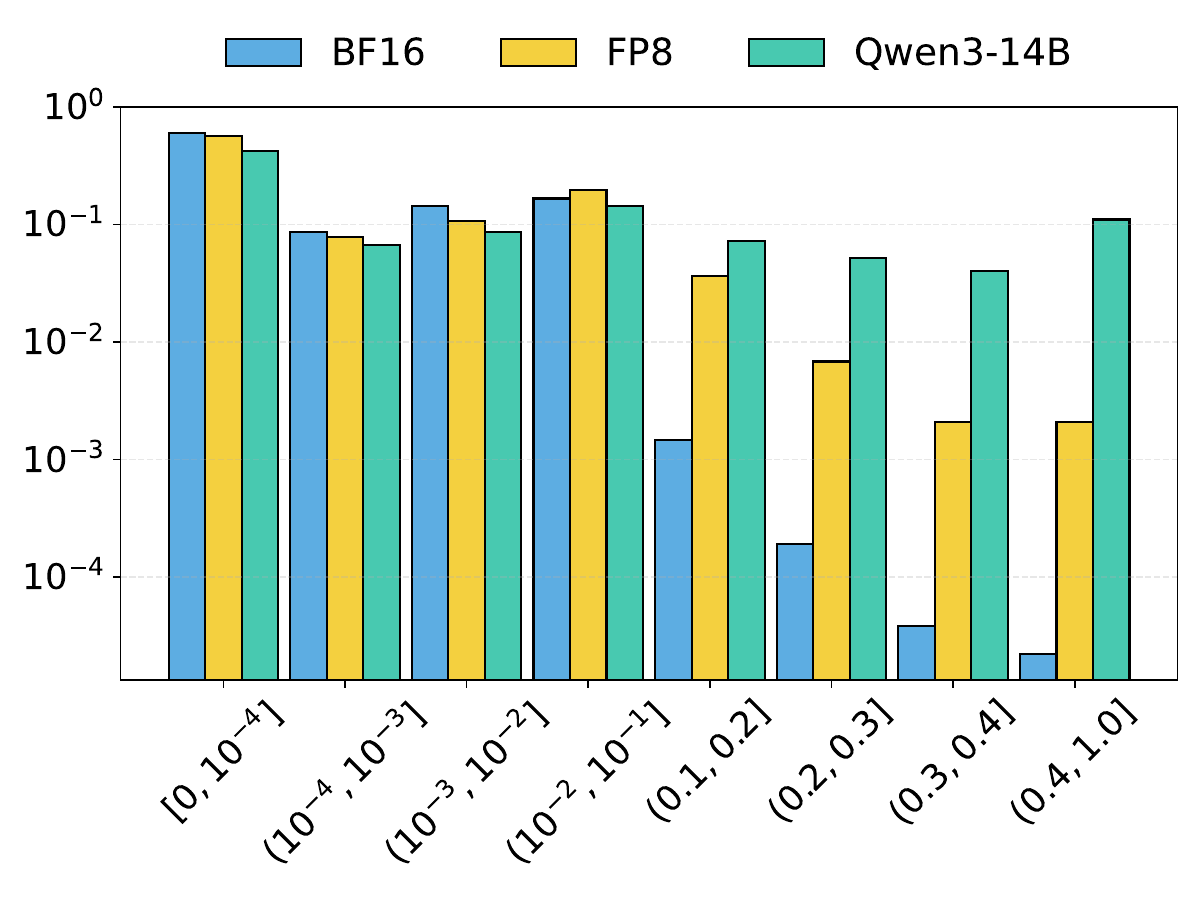}
        \caption{Qwen3-32B}
    \end{subfigure}

    \begin{subfigure}[t]{0.49\linewidth}
        \centering
        \includegraphics[width=\linewidth]{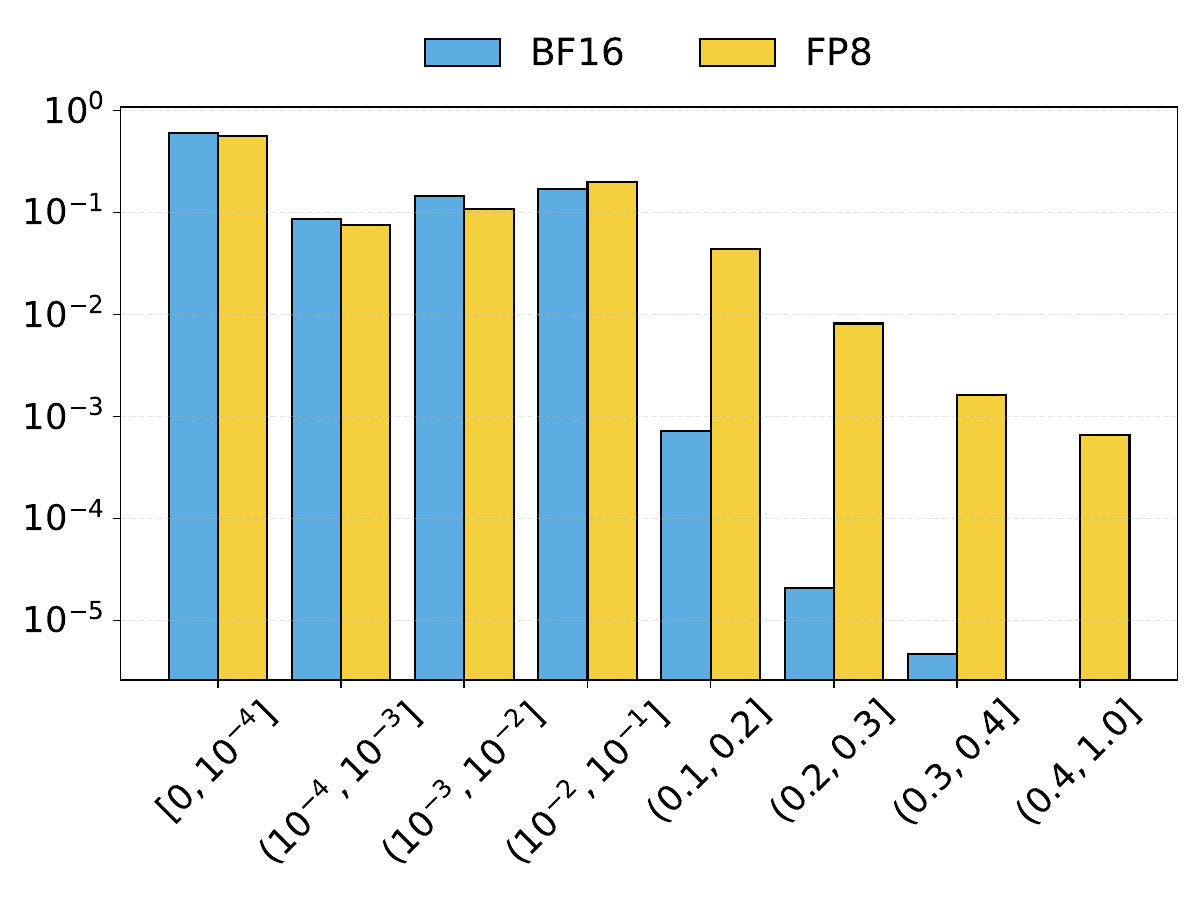}
        \caption{Qwen3-30B-A3B}
    \end{subfigure}
    \hfill
    \begin{subfigure}[t]{0.49\linewidth}
        \centering
        \includegraphics[width=\linewidth]{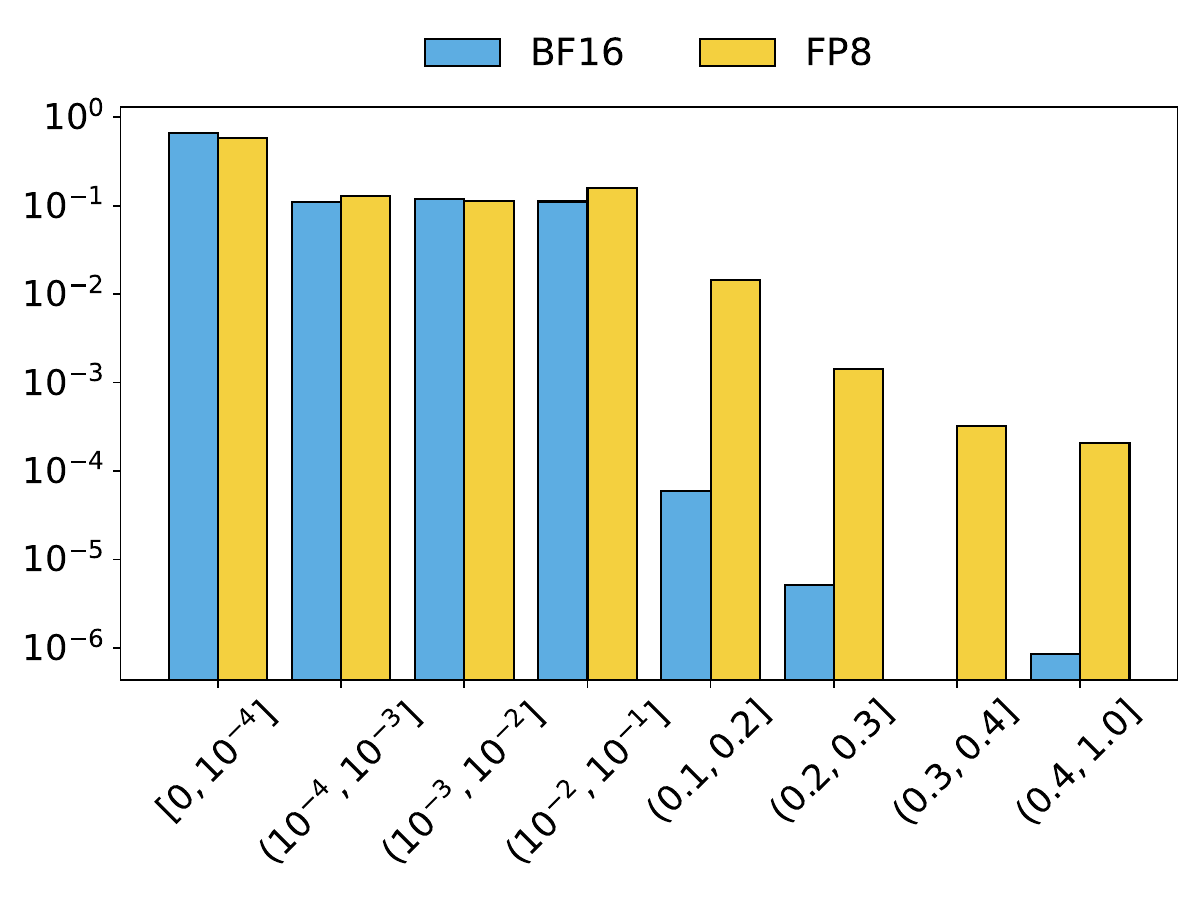}
        \caption{DeepSeek-V2-Lite}
    \end{subfigure}

    \caption{Global logit TV-distance distribution. Probabilities are displayed on a logarithmic y-axis to better capture the tail behavior. }
    \label{fig:global_tv}
\end{figure}

\begin{figure}[ht]
    \centering
    \begin{subfigure}[t]{0.49\linewidth}
        \centering
        \includegraphics[width=\linewidth]{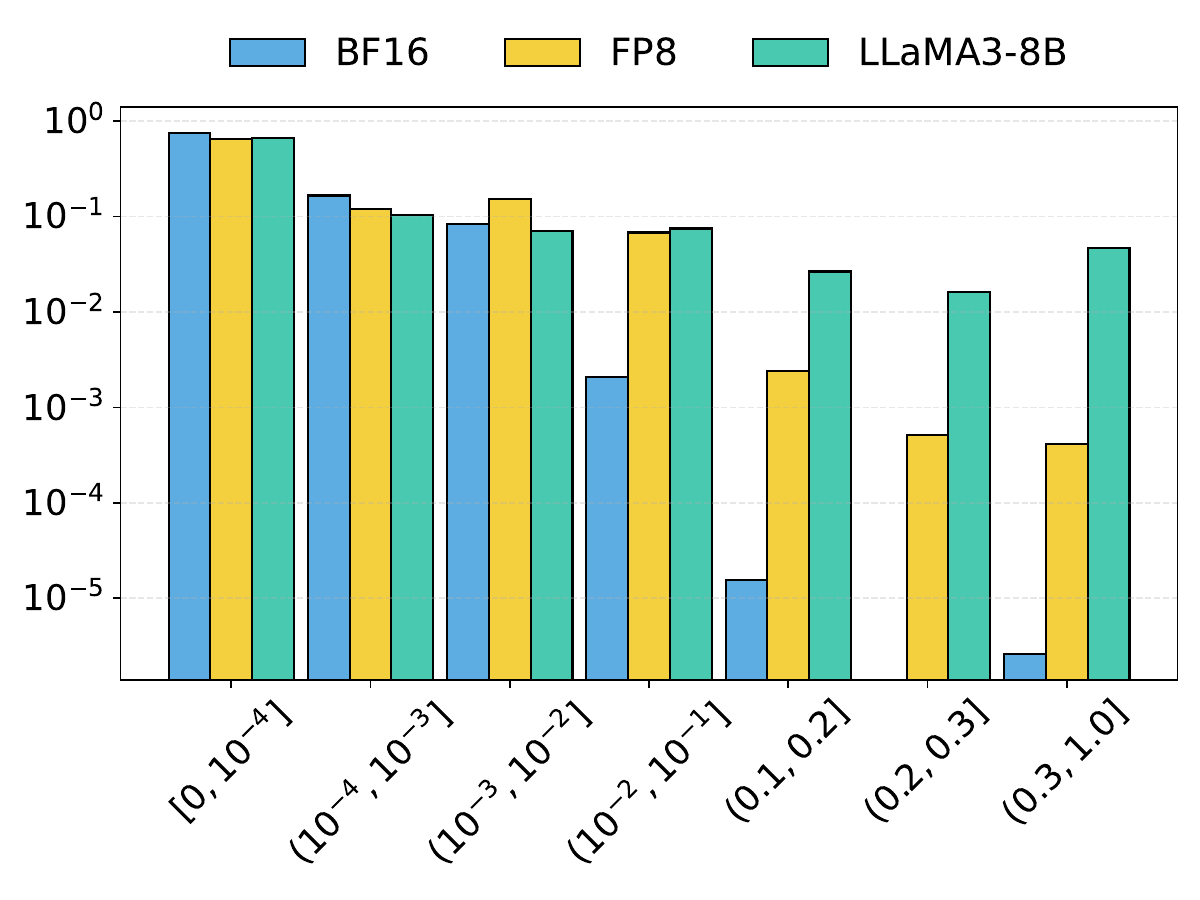}
        \caption{LLaMA3-70B}
    \end{subfigure}
    \hfill
    \begin{subfigure}[t]{0.49\linewidth}
        \centering
        \includegraphics[width=\linewidth]{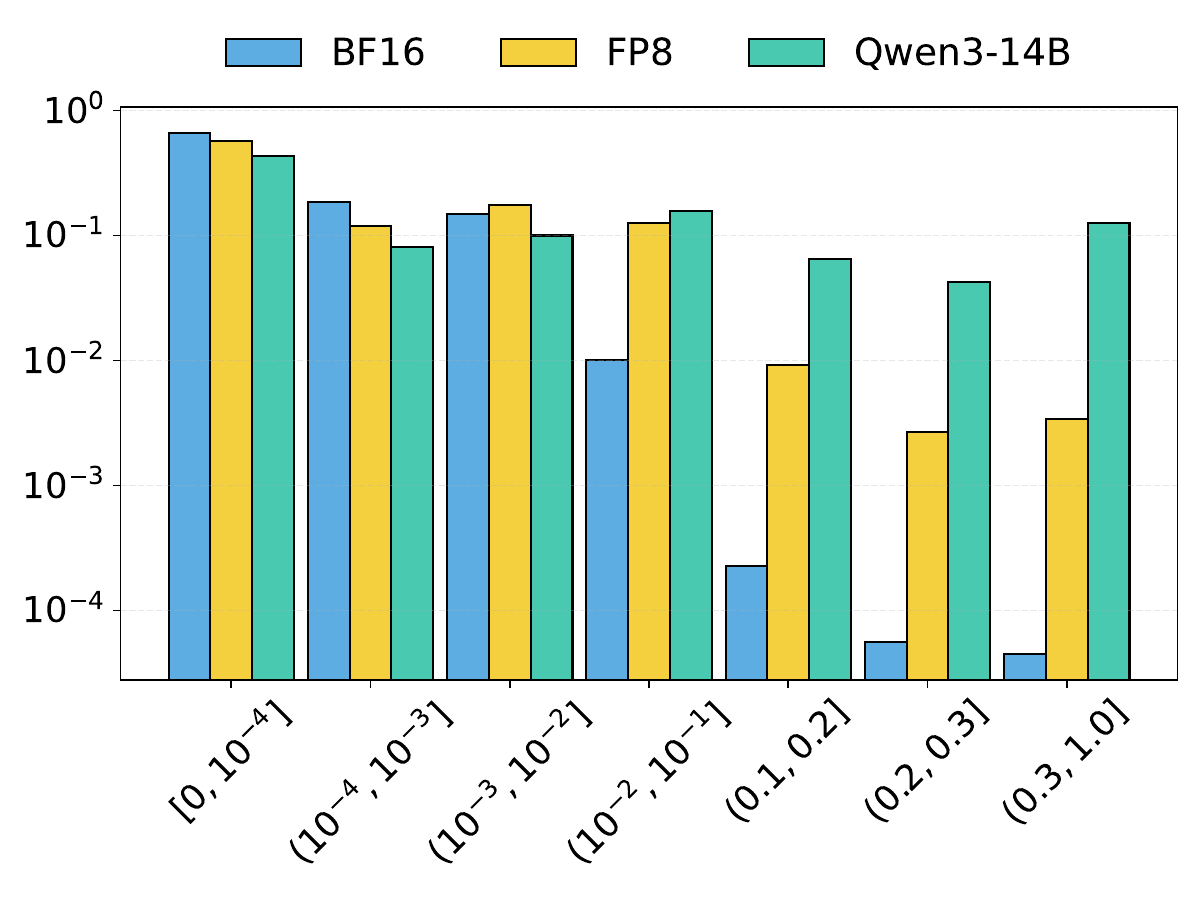}
        \caption{Qwen3-32B}
    \end{subfigure}

    \begin{subfigure}[t]{0.49\linewidth}
        \centering
        \includegraphics[width=\linewidth]{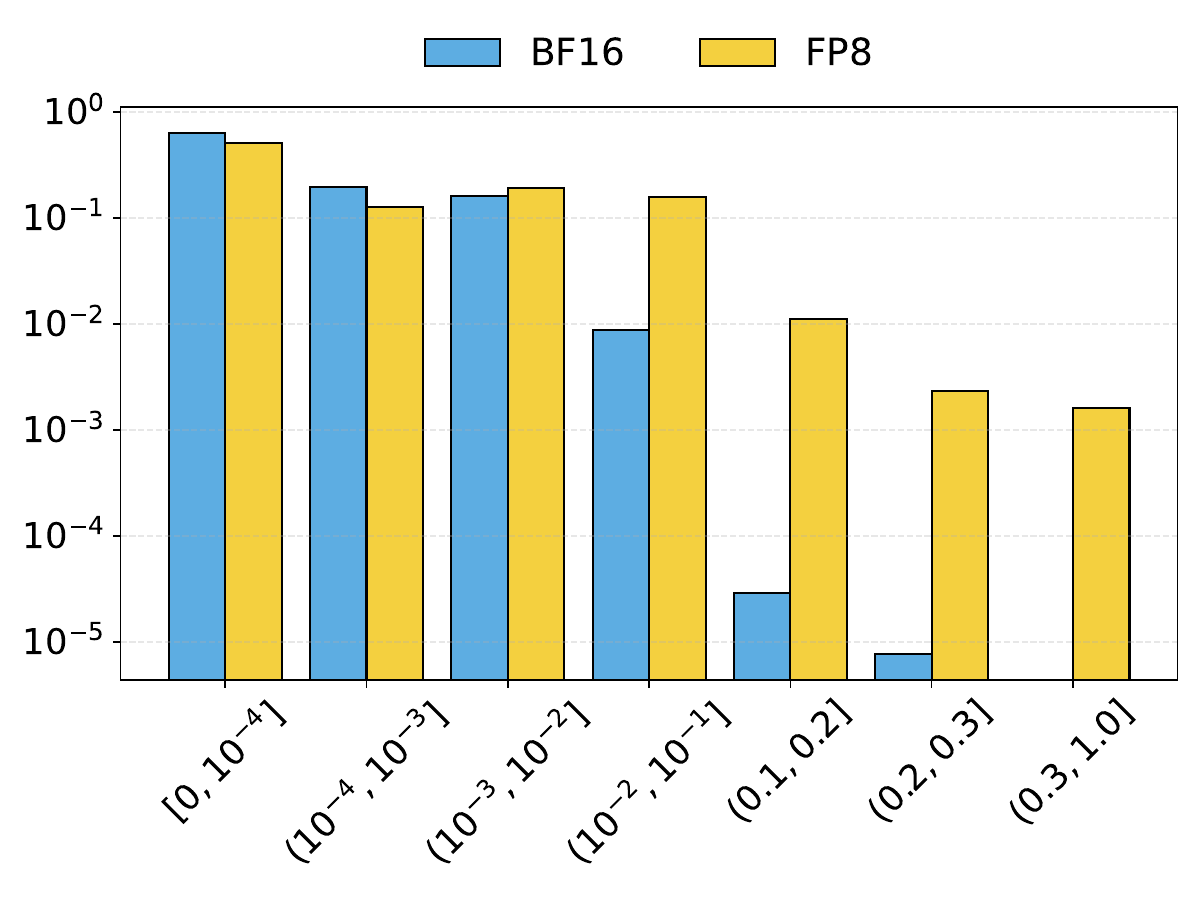}
        \caption{Qwen3-30B-A3B}
    \end{subfigure}
    \hfill
    \begin{subfigure}[t]{0.49\linewidth}
        \centering
        \includegraphics[width=\linewidth]{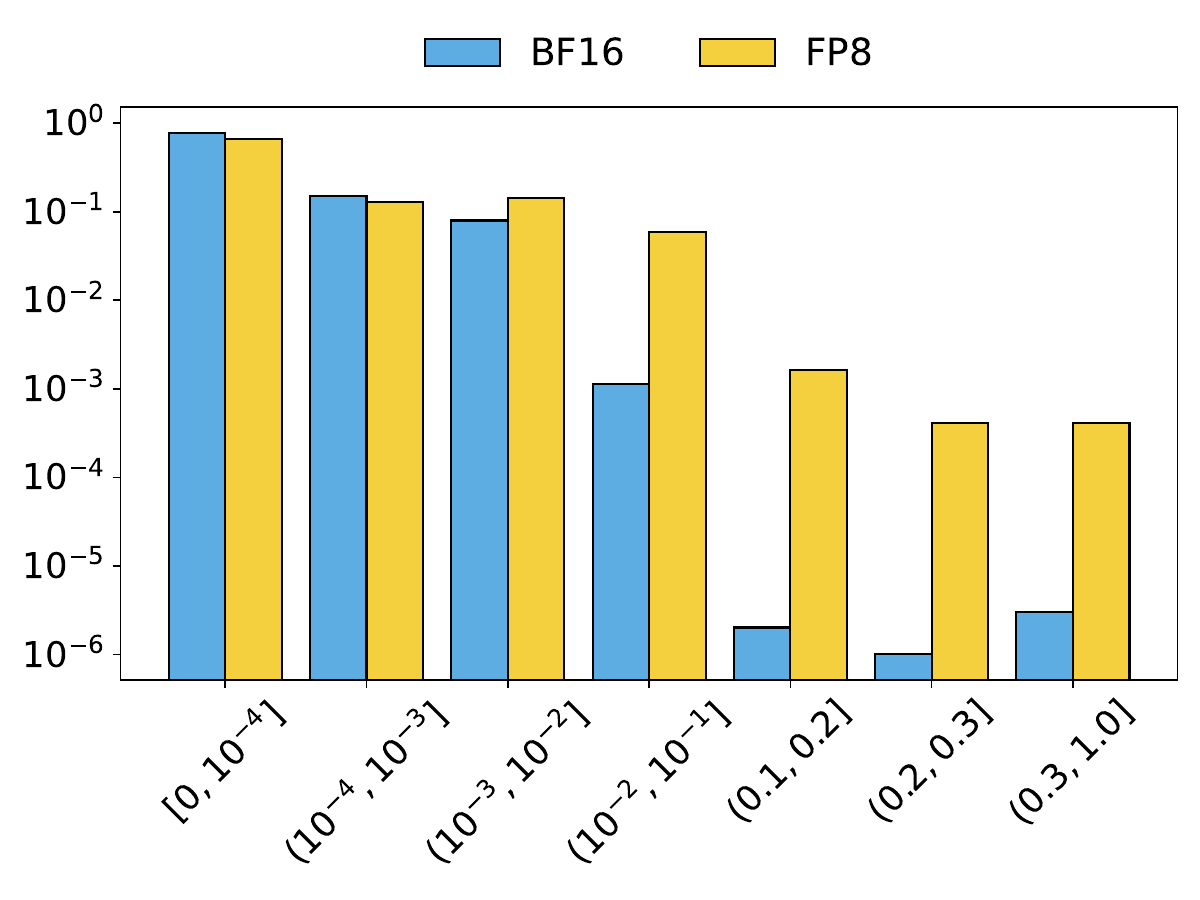}
        \caption{DeepSeek-V2-Lite}
    \end{subfigure}

    \caption{Global logit KL divergence distribution. Probabilities are displayed on a logarithmic y-axis to better capture the tail behavior. }
    \label{fig:global_kl}
\end{figure}

\begin{figure}[ht]
    \centering
    \begin{subfigure}[t]{0.49\linewidth}
        \centering
        \includegraphics[width=\linewidth]{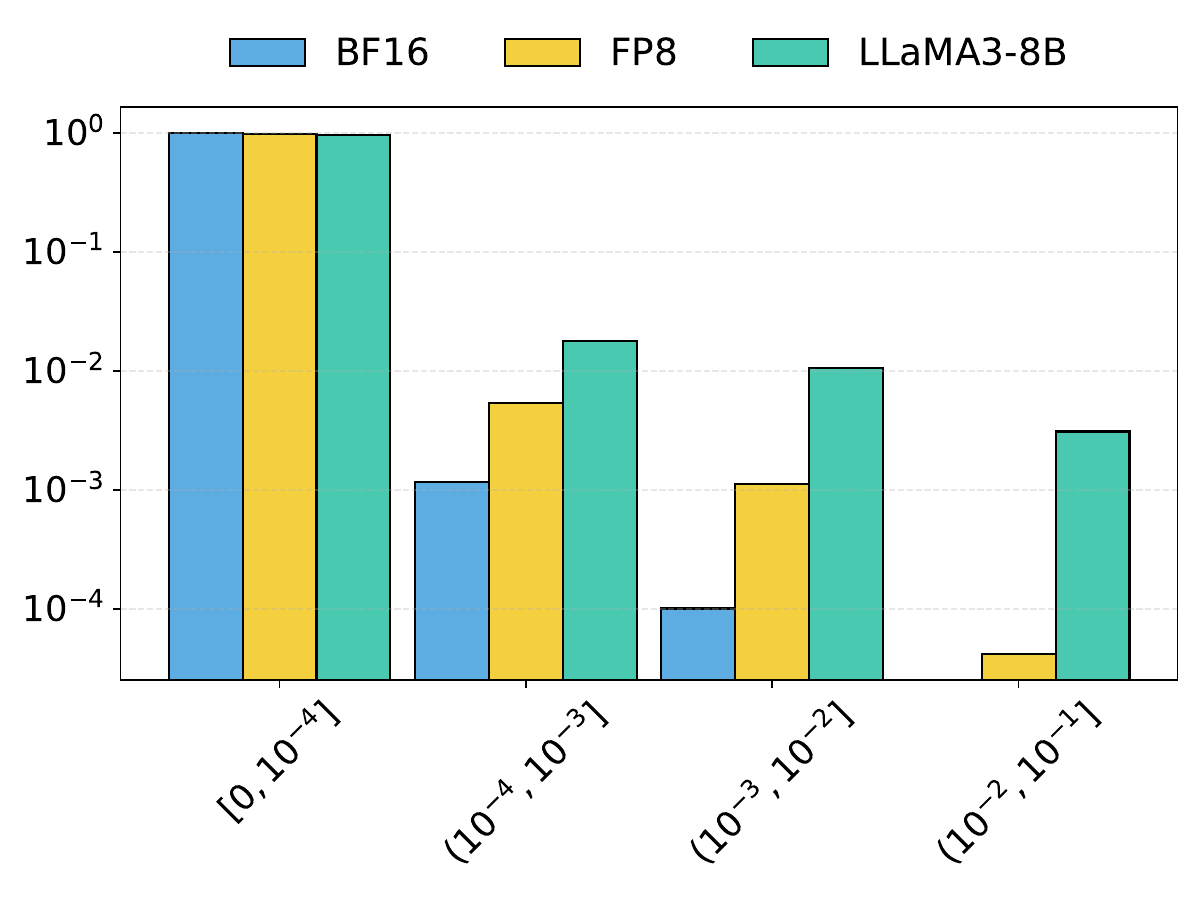}
        \caption{LLaMA3-70B}
    \end{subfigure}
    \hfill
    \begin{subfigure}[t]{0.49\linewidth}
        \centering
        \includegraphics[width=\linewidth]{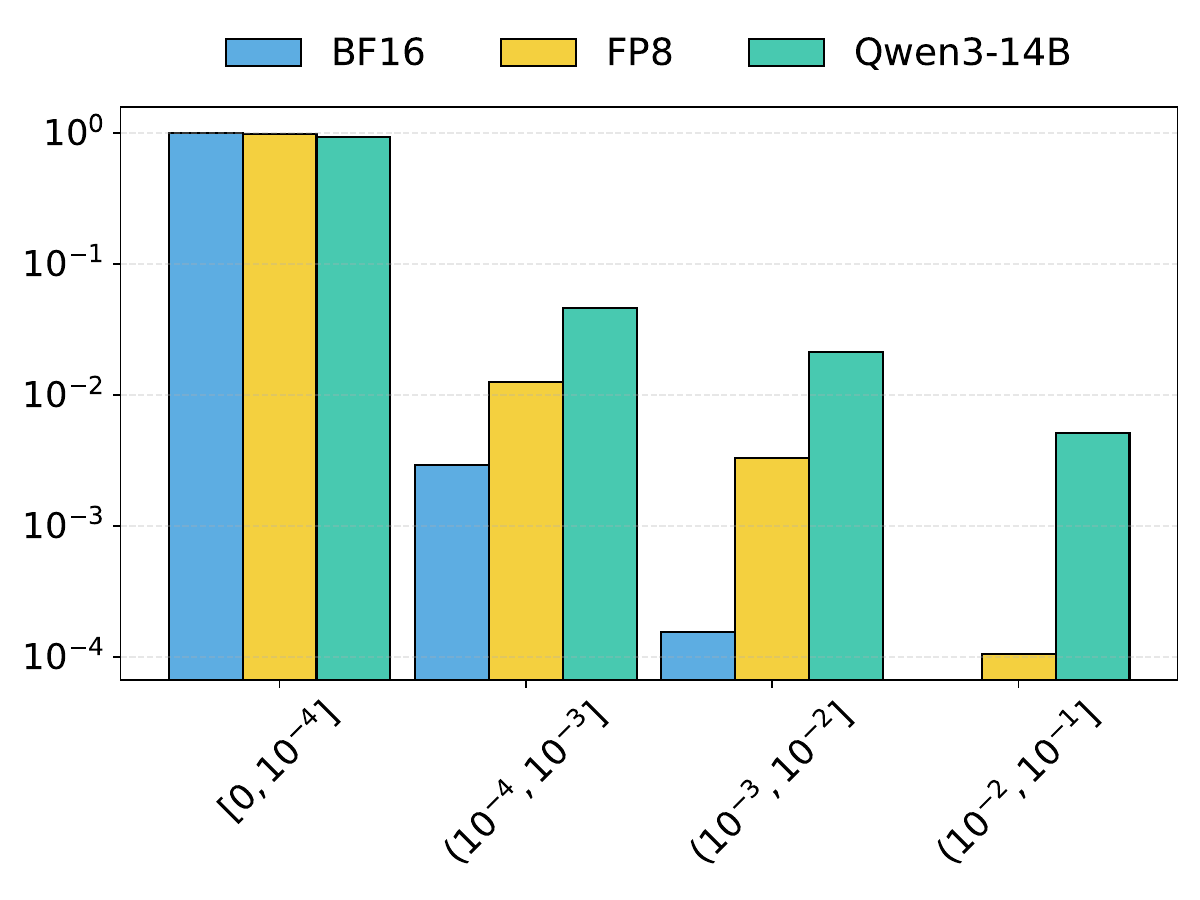}
        \caption{Qwen3-32B}
    \end{subfigure}

    \begin{subfigure}[t]{0.49\linewidth}
        \centering
        \includegraphics[width=\linewidth]{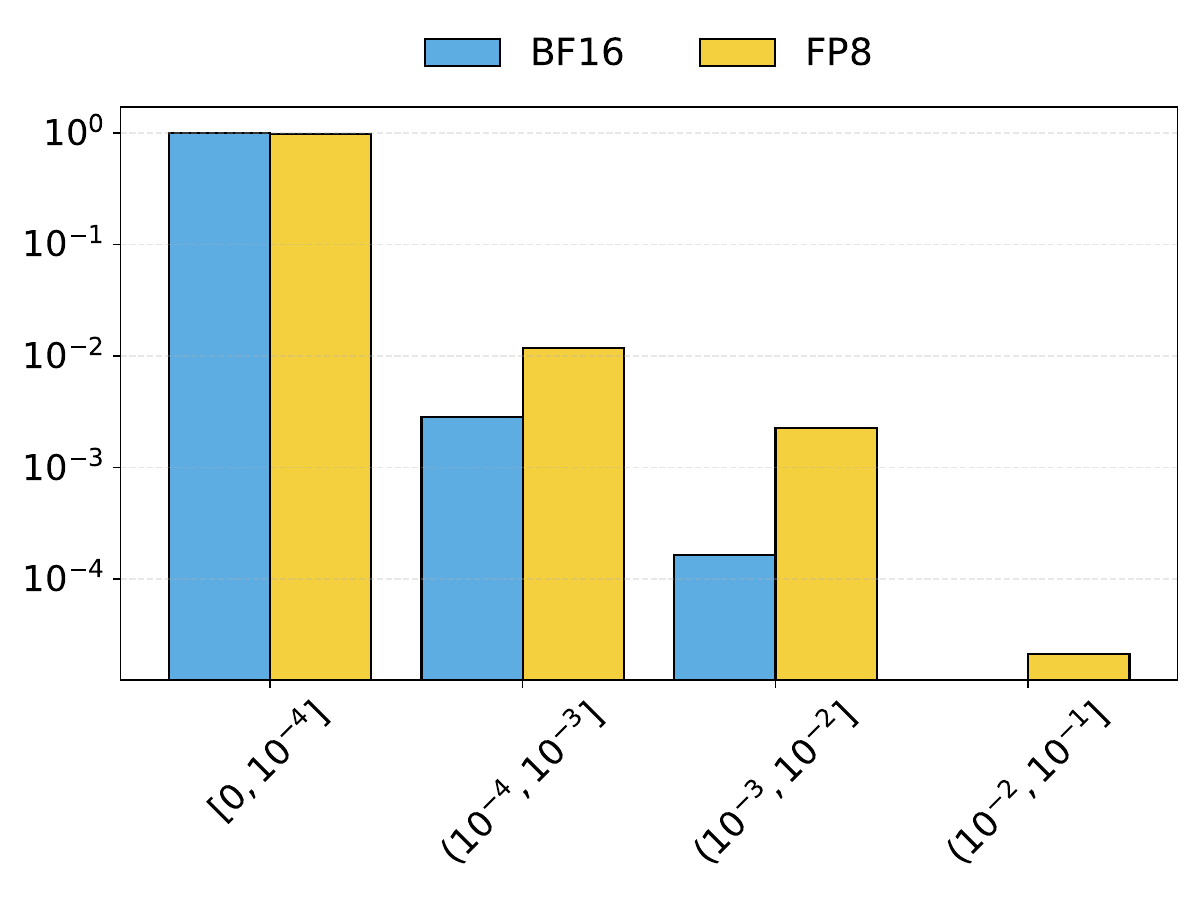}
        \caption{Qwen3-30B-A3B}
    \end{subfigure}
    \hfill
    \begin{subfigure}[t]{0.49\linewidth}
        \centering
        \includegraphics[width=\linewidth]{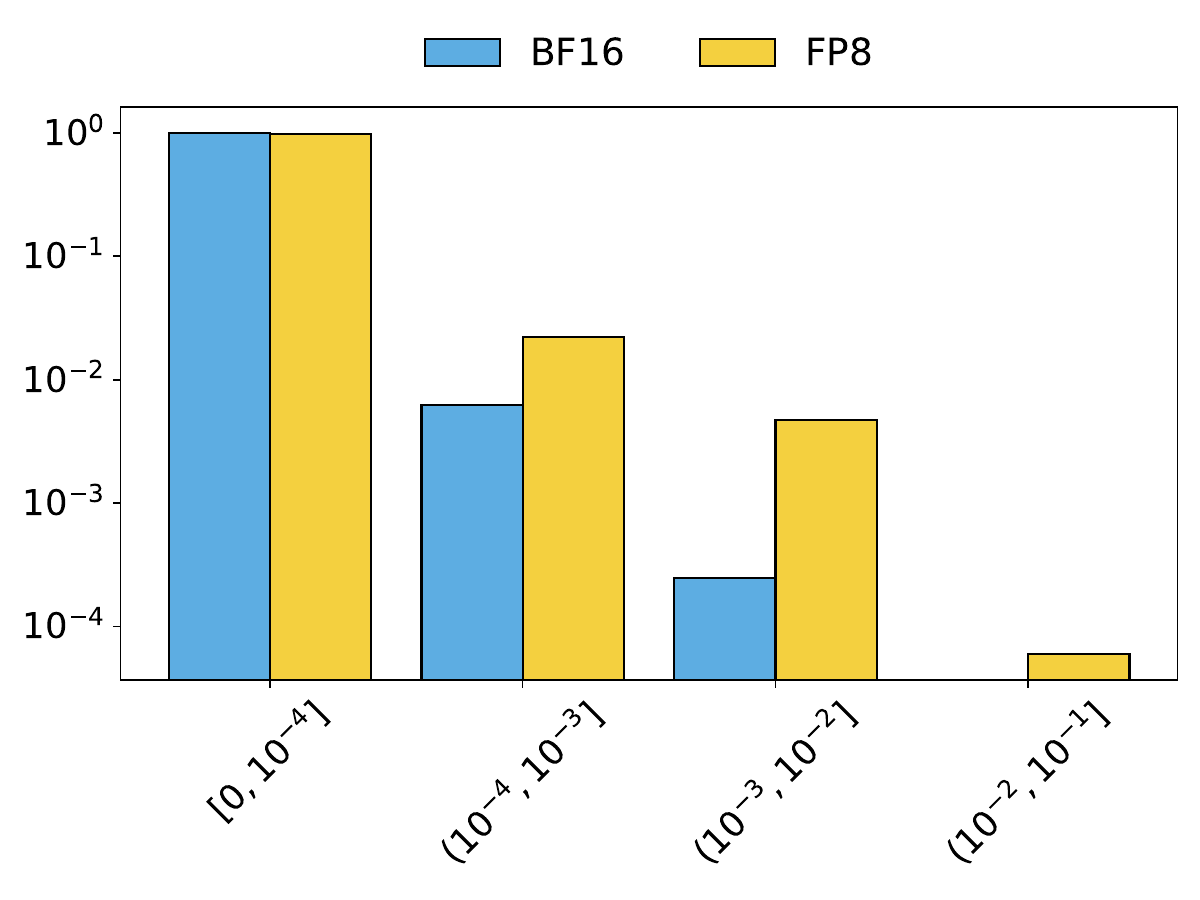}
        \caption{DeepSeek-V2-Lite}
    \end{subfigure}

    \caption{Global logit Top-$K$ distance distribution. Probabilities are displayed on a logarithmic y-axis to better capture the tail behavior. }
    \label{fig:global_topk_diff}
\end{figure}

\begin{figure}[ht]
    \centering

    \begin{subfigure}[t]{0.49\linewidth}
        \centering
        \includegraphics[width=\linewidth]{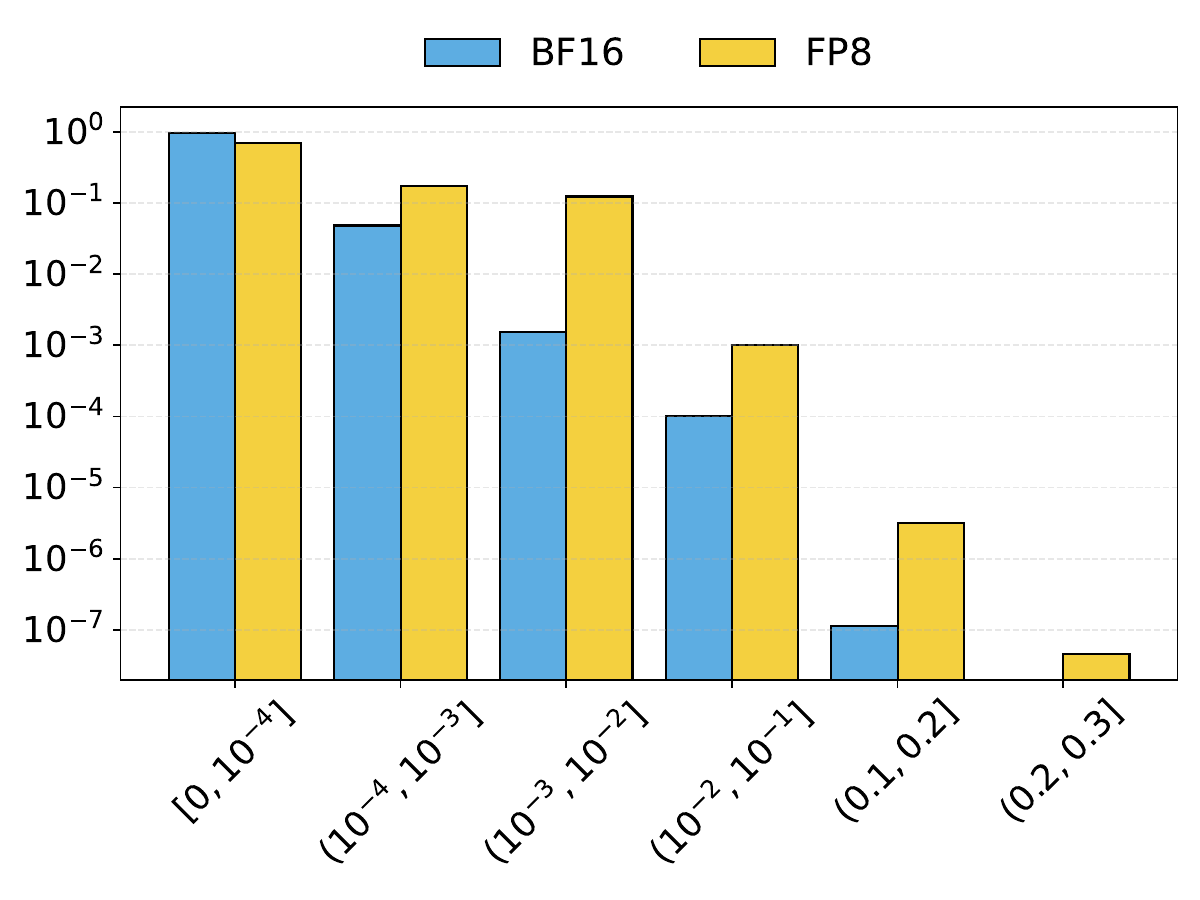}
        \caption{Qwen3-30B-A3B}
    \end{subfigure}
    \hfill
    \begin{subfigure}[t]{0.49\linewidth}
        \centering
        \includegraphics[width=\linewidth]{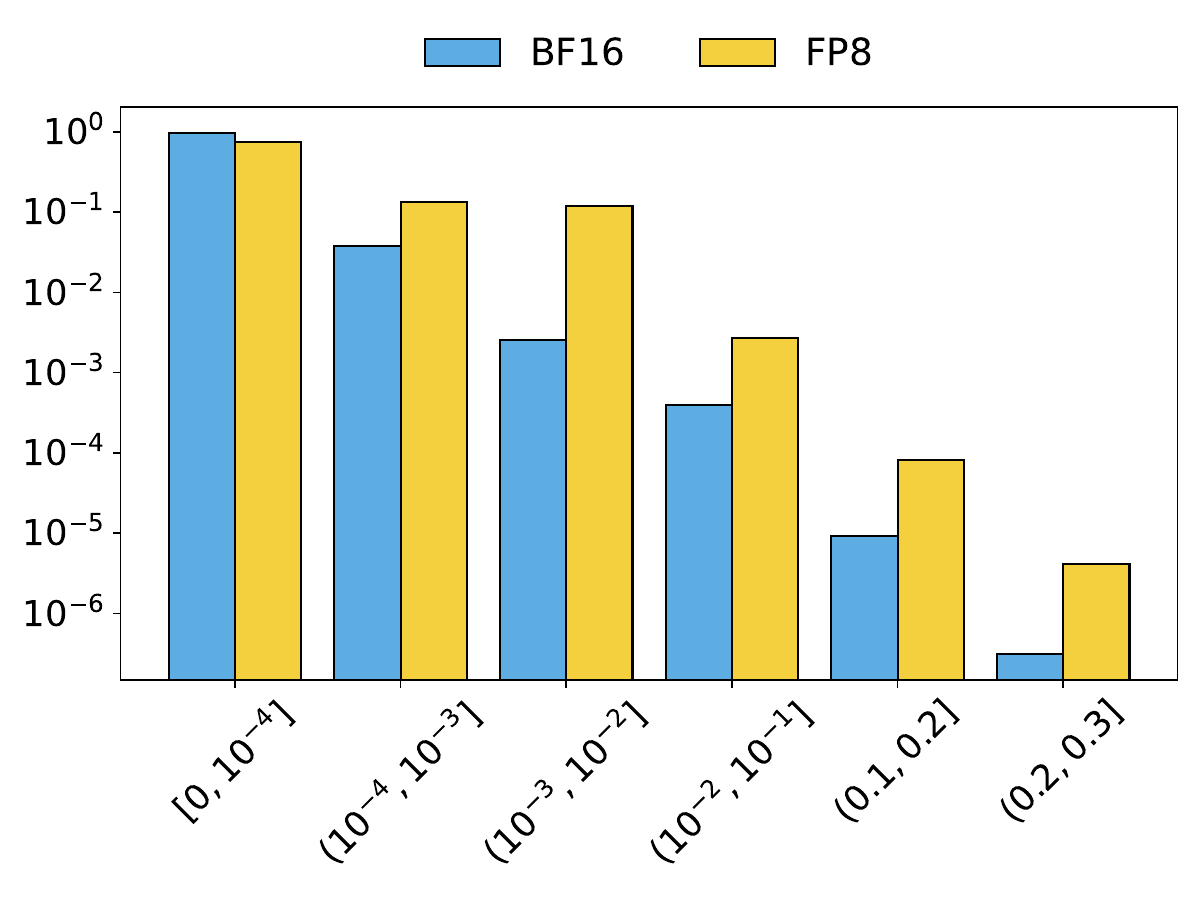}
        \caption{DeepSeek-V2-Lite}
    \end{subfigure}

    \caption{Global logit Top-$K$ distance distribution. Probabilities are displayed on a logarithmic y-axis to better capture the tail behavior. }
    \label{fig:global_expert_diff}
\end{figure}


\end{document}